\documentclass[reqno, leqno,11pt]{article}
\usepackage{amsmath,amssymb,amsthm,amsbsy}
\usepackage[english]{babel}
\usepackage[ansinew]{inputenc}
\usepackage{graphicx}
\usepackage{a4wide}  
\usepackage{hyperref}
\usepackage{xcolor}
\usepackage{bbm,  mathrsfs}
\usepackage{scalerel,stackengine}

\numberwithin{equation}{section}

\newcommand\blue[1]{\textcolor{black}{#1}}
\newcommand{\red}[1]{\textcolor{black}{#1}}

\newtheorem{theo}{Theorem}[section]
\newtheorem{thm}[theo]{Theorem}
\newtheorem{prop}[theo]{Proposition}
\newtheorem{lemma}[theo]{Lemma}
\newtheorem{kor}[theo]{Corollary}
\newtheorem{definition}[theo]{Definition}
\newtheorem{claim}[]{Claim}[section]

\theoremstyle{remark}
\newtheorem{rem}{Remark}[section]

\newtheorem{ex}{Example}[section]

\newcommand{\e}{\mathrm{e}} 
\newcommand{\N}{\mathbb{N}}
\newcommand{\R}{\mathbb{R}}
\newcommand{\Z}{\mathbb{Z}}
\newcommand{\C}{\mathbb{C}}

\newcommand{\E}{\mathbb{E}}
\newcommand{\dd}{\mathrm{d}} 
\newcommand{\eps}{\varepsilon}

\stackMath
\newcommand\reallywidehat[1]{%
\savestack{\tmpbox}{\stretchto{%
  \scaleto{%
    \scalerel*[\widthof{\ensuremath{#1}}]{\kern.1pt\mathchar"0362\kern.1pt}%
    {\rule{0ex}{\textheight}}
  }{\textheight}%
}{2.4ex}}%
\stackon[-6.9pt]{#1}{\tmpbox}%
}
\newcommand{\vect}[1]{\boldsymbol{#1}}

\newcommand{\be}{\begin{equation}}
\newcommand{\ee}{\end{equation}}

\def\1{{\mathchoice {1\mskip-4mu\mathrm l}      
{1\mskip-4mu\mathrm l}
{1\mskip-4.5mu\mathrm l} {1\mskip-5mu\mathrm l}}}

\title{Cluster expansions: Necessary and sufficient convergence conditions}
\author{Sabine Jansen\thanks{Mathematisches Institut, Ludwig-Maximilians-Universit{\"a}t, Theresienstr. 39, 80333 M{\"u}nchen, Germany.}\ \thanks{Munich Center for Quantum Science and Technology (MCQST), Schellingstr. 4, 80799 M{\"u}nchen.} \and Leonid Kolesnikov\footnotemark[1]}
\date{}
\begin{document}




\maketitle

\begin{abstract}	
We prove a new convergence condition for the activity expansion of correlation functions in equilibrium statistical mechanics with possibly negative pair potentials. For non-negative pair potentials, the criterion is an if and only if condition. The condition is formulated with a sign-flipped Kirkwood-Salsburg operator and known conditions such as Koteck{\'y}-Preiss and Fern{\'a}ndez-Procacci are easily recovered. In addition, we deduce new sufficient convergence conditions for hard-core systems in $\R^d$ and $\Z^d$ as well as for abstract polymer systems. \red{The latter improves on the Fern{\'a}ndez-Procacci criterion.}  \\

	
\end{abstract}

\tableofcontents

\section{Introduction}
Since its introduction by Mayer in the early 40s, the method of cluster expansions was --- and remains --- a very important tool in equilibrium statistical mechanics. A classical application yields the analyticity of the logarithm of the partition function for a physical system at equilibrium by deriving a Taylor expansion in the activity or density parameter around zero. Such results can be quite useful, for example, in the study of phase transitions or  the decay of correlations, for a vast class of models. 

In 1971, Gruber and Kunz introduced in their seminal paper~\cite{gruber-kunz1971} systems of non-overlapping geometric objects --- referred to as \emph{polymers} --- given by subsets of a lattice. They presented a rigorous mathematical formalism in order to provide convergent cluster expansions for this model. Instead of the logarithm of the partition function, they considered the correlation functions of the system and derived convergent activity expansions by using a system of integral equations, the so-called Kirkwood-Salsburg equations, and solving the corresponding fixed point equation on a suitable Banach space. However, in the following years less analytical appoaches were favoured by researchers: Combinatorial proofs such as in \cite{brydges1986clustercourse}, relying on tree-graph identities~\cite{o.penrose1967}, and inductive proofs following the idea by Koteck{\'y} and Preiss \cite{kotecky-preiss1986} and its development in~\cite{dobrushin1996} by Dobrushin. The inductive method was presented in the more general setup of abstract polymers \red{(where the underlying space is not necessarily a lattice, nor are the polymers necessarily given by geometric objects)}. Notice that abstract polymer models are universal in the sense that a large class of classical models can be represented as polymer models due to the combinatorial structure of the corresponding partition functions (see, e.g.,~\cite{friedli-velenik2018book} for an application to the Ising model). Moreover, an interesting connection with probability theory was pointed out by Scott and Sokal in~\cite{scott-sokal2005}: Convergence of cluster expansions in abstract polymer models is related to the Lov{\'a}sz Local Lemma --- better sufficient conditions can provide refinements of the latter (see, e.g.,~\cite{bissacot2011improv}).

In 2008, Fern{\'a}ndez and Procacci proved a new sufficient criterion in the setup of abstract polymers improving on the result by Koteck{\'y} and Preiss. The initial proof~\cite{fernandez-procacci2007} relies on combinatorial arguments, an alternative proof via an induction \`{a} la Dobrushin~\cite{Fialho2020AbstractPG} appeared recently (finally, in this paper we provide an analytical proof in the spirit of Gruber-Kunz). 

Overall, in the last two decades, a notable effort was made to generalize classical sufficient conditions in the abstract polymer setup (including the condition by Fern{\'a}ndez and Procacci) to hold in continuous spaces and for systems with soft-core (or even more general) interactions, see~\cite{ueltschi2004,poghosyan-ueltschi2009,faris2008,jansen2019clustergibbs,fernandez-nguyen2020}. 

\blue{We want to go further by} employing a Kirkwood-Salsburg approach in the rather general setup of Gibbs point processes (or, in terms of statistical mechanics, grand-canonical Gibbs measures) defined via pairwise interactions. It is well-known that --- under mild additional moment conditions, which are automatically satisfied for non-negative pair potentials --- there is a one-to-one correspondence between the set of those Gibbs measures and the associated families of correlation functions (also known as \emph{factorial moment densities}). In the special case of a discrete space and hard-core interactions, the value of the $n$-point correlation function is given simply by the probability to see $n$ particles at the prescribed positions in the random configuration of particles. The correlation functions can be expanded as power series in the activity parameter $z$, i.e., in the intensity of the underlying Poisson process. We denote the Taylor expansion for the $n$-point correlation function in $z$ around zero by $\rho_n$ and write $\vect \rho$ for the family of those expansions. In general, the series $\vect \rho$ need not to be convergent at all; we are, however, interested in conditions which ensure pointwise convergence (towards the correlation functions). Furthermore, we want to consider the more general case where the underlying Poisson point process is inhomogeneous, i.e., where a different intensity value may be assigned to every point in the space, the activity $z$ is a function and the expansions $\vect \rho$ are multivariate power series in $z$. For a rigorous introduction of Gibbs point processes and corresponding correlation functions, see~\cite{jansen2019clustergibbs}, but notice that here we do not assume the interaction potential to be non-negative (unless explicitly stated).

The starting point of the paper and the central quantity to investigate are the activity expansions $\vect \rho$ which we consider independently of their interpretation in term of the correlation functions. Let us outline the main ideas present in the paper. The coefficients of the multivariate power series $\vect \rho$ are defined in terms of a certain family of rooted graphs \red{to which we refer as \emph{multi-rooted graphs} (see \cite{stell1964, jansen2019clustergibbs})}. Using the terminology from~\cite{faris2010species}, the activity expansions $\vect \rho$ are given by the exponential generating functions of the coloured weighted combinatorial species of multi-rooted graphs with a fixed set of roots. The set of all multi-rooted graphs has an essential structural property --- it is invariant under the operation of removal of a root. Taking a multi-rooted graph and removing an arbitrary root (as well as all edges incident to it), one gets again a multi-rooted graph on a smaller vertex set, where every neighbour of the removed root becomes a root vertex itself. The weight of the original graph is equal to the weight of the resulting graph times the weight of the edges removed. The corresponding property of the generating functions is expressed by the Kirkwood-Salsburg equations. Every possible rule for the choice of the root to remove induces a different combinatorial operation and therefore a different system of Kirkwood-Salsburg equations for the generating functions.

 \red{In this work we provide a condition for absolute convergence of the activity expansions $\vect \rho$ in terms of the existence of a measurable function solving a system of Kirkwood-Salsburg type inequalities (in the case of repulsive interactions, that condition is also a necessary one).} Our main result, Theorem~\ref{thm:main1}, is inspired by~\cite{bissacot-fernandez-procacci2010}; it is a slightly modified, strongly generalized version of Claim 1 therein. The goal, however, is not only to obtain abstract conditions which are both necessary and sufficient for convergence of the cluster expansions --- but also to demonstrate how these characterizations provide a universal approach to prove model-specific sufficient conditions on different levels of generality, both in discrete and continuous setups with repulsive interactions. A two-lane mechanism arises: On the one hand, for a candidate family of ansatz functions $\vect \xi$ (given, for example, as approximations of $\vect \rho$) one can search for conditions that ensure that these functions $\vect \xi$ satisfy the Kirkwood-Salsburg inequalities; on the other hand, given candidate sufficient conditions, one can construct a suitable family of ansatz functions $\vect \xi$ tailored to satisfy the Kirkwood-Salsburg inequalities under these conditions. 

This approach provides a unifying framework for the known conditions, but it also allows to prove stronger results. To emphasize this possibility, we derive a new sufficient condition for absolute convergence of the activity expansions \red{$\vect \rho$} in the setup of abstract polymers. In that general setup, our condition improves on any known condition that we are aware of. 

A more detailed outline of the main ideas intoduced above   can be found in~\cite{jansen2020proc} (for the case of non-negative pairwise interactions and without rigorous proofs).

In the further course of the paper, we investigate two particular hard-core setups as examples --- the subset polymers in $\Z^d$ and hard objects in $\R^d$.  There, the sets of roots of the multi-rooted graphs correspond to configurations of geometric objects. By breaking the geometric objects into smaller ``pieces" to which we refer as \emph{snippets}, we can identify these configurations with configurations of snippets (e.g., in the case of subset polymers we can identify a configuration of polymers with the disjoint union of monomers covering this configuration). Picking a root of a multi-rooted graph --- the combinatorial operation underlying the Kirkwood-Salsburg equations --- corresponds to picking a snippet. Different rules to pick a snippet in general give rise to characterizations of absolute convergence in terms of different Kirkwood-Salsburg inequalities. This way the latter can be tailored to a candidate  sufficient condition. Thus different sufficient conditions can be derived by playing both with the choice of different systems of  Kirkwood-Salsburg inequalities and the choice of different ansatz functions satisfying these inequalities. We illustrate this mechanism by deriving some sufficient conditions for a class of hard-core interaction models, in particular for multi-type systems of hard spheres in $\R^d$.  

The paper is organized as follows: In Section~\ref{sec2.1} we introduce the basic notation and present the general framework. Furthermore, in Theorem~\ref{thm:main1} we state our main result, a characterization of the domain of absolute convergence for the activity expansions $\vect \rho$, and use it to recreate the classical sufficient conditions by Koteck{\'y} and Preiss as well as the sufficient conditions by Fern{\'a}ndez and Procacci in a rather general setup (see  Corollaries~\ref{KPcond} and ~\ref{FPcond}, respectively). In Section~\ref{sec:abstract-polymers}, the same approach is used to prove a new, improved sufficient condition in the setup of abstract polymers (Proposition~\ref{goodprop}). The proof of the proposition relies on an auxiliary result (Lemma~\ref{kl2}) which is proved in Appendix~\ref{AppB}. In the Sections~\ref{sec:hardcore} and~\ref{sec:sub_poly} we consider the special case of hard-core interactions. Both in the continuum (Section~\ref{sec:hardcore}) and in the discrete setup (Section~\ref{sec:sub_poly}), we provide model-specific characterizations of the convergence domain, stated in Theorem~\ref{thm:main2} and Theorem~\ref{thm:bfp}, respectively. As an immediate consequence of Theorem~\ref{thm:bfp} we obtain an elementary proof of the well known Gruber-Kunz condition (Corollary~\ref{cor:GK}). In Section 3, we present a forest-graph equality and other combinatorial results in order to prove Theorems~\ref{thm:main1},~\ref{thm:main2} and~\ref{thm:bfp}. Finally, in Section 4, Theorem \ref{thm:main2} and Theorem~\ref{thm:bfp} are used to obtain practitioner-type sufficient conditions for a class of hard-core interaction models, including new sufficient conditions for subset polymers in $\Z^d$ (Theorem~\ref{Lgoof}) and hard objects in $\R^d$ (Theorem~\ref{cLgoof} and Theorem~\ref{GK_cont}). 

\red{The reader interested primarily in the discrete setup of subset polymers is encouraged to jump directly to Subsection~\ref{sec:sub_poly}, its main result being the characterization of the domain of convergence for the activity expansions $\vect \rho$ given by Theorem~\ref{thm:bfp} (compare to Theorem~\ref{thm:gen_bfp}). The main ideas behind the proof of Theorem~\ref{thm:bfp} in Subsection~\ref{sec:recc_sub} and behind the application of Theorem~\ref{thm:bfp} in Subsection 4.1 can be transferred to the continuous setup as well.}

\section{Main results} 
\subsection{(Locally) stable pair potentials}
\label{sec2.1}

Let $(\mathbb X, \mathcal X)$ be a measurable space, $\lambda$ a $\sigma$-finite reference measure, and $v$ a pair potential, i.e., 
$v: \mathbb X\times \mathbb X\to \R\cup\{\infty\}$ is measurable and symmetric --- in the sense that $v(x,y) = v(y,x)$ for any $x,y\in\mathbb{X}$. Corresponding to the potential $v$, \emph{Mayer's $f$ function}  is  given by
$$
	f(x,y) = \e^{-v(x,y)}-1.
$$

We call the pair potential $v$ \emph{stable} if there exists a measurable map $B:\mathbb X\to \R_+$ such that for any $n\in\mathbb{N}$ and $x_1,\ldots,x_n\in\mathbb{X}$
\begin{align}\label{def:stab}
\prod\limits_{1\leq i<j\leq n}(1+f(x_i,x_j))\leq \e^{\sum_{k=1}^nB(x_k)}
\end{align}
holds; we call $v$ \emph{locally stable} or \emph{Penrose stable} (due to O. Penrose, see \cite{o.penrose1967}) if  there exists a measurable map $C:\mathbb X\to \R_+$ such that for any $x_0\in\mathbb{X}$, $n\in\mathbb{N}$ and $x_1,...,x_n\in\mathbb{X}$ satisfying $\prod_{1\leq i<j\leq n}(1+f(x_i,x_j))\neq 0$ 
\begin{align}\label{def:loc_stab}
\prod\limits_{i=1}^n(1+f(x_0,x_i))\leq\e^{C(x_0)}
\end{align}
holds. Notice that every locally stable potential is stable and that every non-negative potential $v$ is locally stable (with the choice $C\equiv0$).

An \emph{activity function} is a measurable map $z:\mathbb X\to \R$. Physically relevant activities are non-negative but for the purpose of studying the convergence of expansions it can be helpful to admit negative (or complex) activities as well. We define the (signed) measure $\lambda_z$ on $\mathcal X$ by 
\begin{align}\label{def:lambda_z}
\lambda_z(B):= \int_B z(x) \lambda(\dd x), \quad B\in\mathcal X.
\end{align}

The weight of a graph $G$ with vertex set $[n]= \{1,\ldots,n\}$ and edge set $E(G)$ is 
$$
	w(G;x_1,\ldots,x_n) := \prod_{\{i,j\}\in E(G)} f(x_i,x_j).
$$
Let $\mathcal G_n$ be the set of all graphs with vertex set $[n]$, $\mathcal C_n\subset\mathcal G_n$ the set of connected graphs  and 
$$
	\varphi_n^\mathsf T(x_1,\ldots x_n):= \sum_{G \in \mathcal C_{n}} w(G;x_1,\ldots,x_n)
$$
the $n$-th Ursell function. For $n\in \N$ and $k\in\N_0$, let $\mathcal D_{n,n+k}\subset \mathcal G_{n+k}$ be the collection of all graphs $G$ such that every vertex $j\in \{n+1,\ldots, n+k\}$ connects to at least one of the vertices $i\in \{1,\ldots, n\}$. We may view the vertices  $\{1,\ldots, n\}$ as roots and call the graphs $G\in \mathcal D_{n,n+k}$ \emph{multi-rooted graphs} or,  following the footnote 53 in~\cite{stell1964}, \emph{root-connected} graphs. 
Consider the functions
$$
		\psi_{n,n+k}(x_1,\ldots,x_{n+k}) := \sum_{G \in \mathcal D_{n,n+k}} w(G;x_1,\ldots,x_{n+k}).
$$
For $n=1$, the functions coincide with the standard Ursell functions, i.e., $\psi_{1,1+k} = \varphi_{1+k}^\mathsf T$. We are interested in the associated series 
$$
		\rho_n(x_1,\ldots,x_n;z) 
		:=\sum_{k=0}^\infty \frac{1}{k!} \int_{\mathbb X^k} \psi_{n,n+k}(x_1,\ldots,x_n,y_1,\ldots,y_k) z(x_1)\cdots z(x_n)  \lambda_z^k(\dd \vect y). 
$$
The summand for $k=0$ is to be read as $\psi_{n,n}(x_1,\ldots,x_n) z(x_1)\cdots z(x_n)$. The series $\rho_n$ corresponds to the $n$-point correlation function of a grand-canonical Gibbs measure \cite[Eq. (4-7)]{stell1964}, see also \cite{jansen2019clustergibbs} --- it is the expansion of the correlation function in the activity $z$ around $0$. 


\blue{We will say that the activity expansions $\vect \rho$ converge absolutely for a non-negative activity function $z$ if
$$
\sum_{k=0}^\infty \frac{1}{k!} \int_{\mathbb X^k}\vert\psi_{n,n+k}(x_1,\ldots,x_n,y_1,\ldots,y_k)\vert z(x_1)\cdots z(x_n)  \lambda_z^k(\dd \vect y)<\infty 
$$
for all $n\in\N$ and $(x_1,\ldots,x_n)\in\mathbb{X}^n$.}

Our main concern is to derive necessary and sufficient convergence conditions, but sometimes it is useful to view the series as purely formal; relevant background on formal power series whose variable is a measure (here $\lambda_z(\dd x)$) is given in~\cite[Appendix~A]{jansen-kuna-tsagkaro2019}. 

Next we introduce sign-flipped Kirkwood-Salsburg \red{operators}. A \emph{selection rule} $s(\cdot)$ is a map from $P(\mathbb X):= \sqcup_{n=1}^\infty \mathbb X^n$ to $\N$ such that $s(x_1,\ldots,x_n) \in \{1,\ldots, n\}$ for all $(x_1,\ldots, x_n) \in P(\mathbb X)$. To lighten notation we write $x_s$ rather than $x_{s(x_1,\ldots,x_n)}$. Further let $(x'_2,\ldots, x'_{n})$ be the vector obtained from $(x_1,\ldots,x_n)$ by deleting the entry $x_s$, leaving the order otherwise unchanged. For the simplest selection rule that picks the first entry $s =1$, we have $x'_i = x_i$. The \emph{sign-flipped Kirkwood-Salsburg operator} $\tilde K_z^s$ with selection rule $s(\cdot)$ acts on families $\vect \xi = (\xi_n)_{n\in \N}$ of measurable symmetric functions $\xi_n:\mathbb X^n\to \R_+$ as 
\begin{multline}
	(\tilde K_z^s \xi)_n(x_1,\ldots,x_n):= z(x_s)\,  \prod\limits_{i=2}^n(1+f(x_s,x'_i))  \Big( \1_{\{n\geq 2\}}\xi_{n-1}(x'_2,\ldots, x'_{n}) \\
		+   \sum_{k=1}^\infty \frac{1}{k!} \int_{\mathbb X^k} \prod_{j=1}^k \bigl| f(x_s, y_j)\bigr|\, \xi_{n-1+k} (x'_2,\ldots, x'_{n}, y_1,\ldots, y_k) \lambda^k(\dd \vect y)\Big),
\end{multline}
for all $n\in\N$ and $(x_1,\ldots,x_n)\in\mathbb X^n$. Here we allow the functions $(\tilde K_z^s \xi)_n$ to assume the value ``$\infty$". For non-negative potentials \red{and on a suitably reduced domain}, $\tilde K_z^s$ differs from the standard Kirkwood-Salsburg operator~\cite[Chapter 4.2]{ruelle1969book} by a mere sign-flip: it has $|f(x_s,y_i)|$ instead of $f(x_s,y_i)$.

\begin{thm}  \label{thm:main1}
	Let $z(\cdot)$ be a non-negative activity and $s(\cdot)$ any selection rule. Consider the following two conditions:
	
	\begin{enumerate} 
		
		\item [(i)] There is a family $\vect \xi = (\xi_n)_{n\in \N}$ of measurable symmetric functions $\xi_n:\mathbb X^n\to \R_+$ such that 
		\be
		\begin{aligned} \label{eq:condition1}
			 z(x_1) \delta_{n,1} +  (\tilde K_z^s \vect \xi)_n\, (x_1,\ldots,x_n) \leq \xi_n(x_1,\ldots,x_n) 
		\end{aligned}
		\ee
		for all $n\in \N$ and  $(x_1,\ldots,x_{n})\in \mathbb X^{n}$.	\end{enumerate}

\begin{enumerate}
		
		 \item [(ii)] The series $\rho_n(x_1,\ldots, x_n;z)$ converges absolutely, for all $n\in \N$ and $(x_1,\ldots, x_n)\in \mathbb X^n$. 
	\end{enumerate} 
	
 	Condition (i) is sufficient for (ii) to hold; moreover, if (i) is satisfied, then 
	\be \label{rhokbound} 
		\sum_{k=0}^\infty \frac{1}{k!} \int_{\mathbb X^k} \bigl| \psi_{n,n+k}(x_1,\ldots,x_n,y_1,\ldots,y_k) \bigr| z(x_1)\cdots z(x_n)  \lambda_z^k(\dd \vect y)	\leq \xi_n(x_1,\ldots, x_n)
	\ee
	on $\mathbb X^n$,
	for all $n\in \N$.
	
	In addition, if we assume the pair potential to be non-negative, then (ii) implies (i) as well, so that the two conditions are equivalent in this case.
\end{thm} 

\begin{rem}We formulate this theorem --- as well as the following results ---
 for non-negative activities, mainly for the purpose of notational convenience. Naturally, such conditions for absolute convergence can be formulated in the usual framework of complex analysis by exchanging complex activities $z$ with $\vert z \vert$ in the convergence criteria. 

\end{rem}

We prove the theorem in Subsection \ref{subsect:3.2.}.
The known sufficient convergence conditions of Koteck\'y-Preiss and Fern\'andez-Procacci types are easily recovered from Theorem~\ref{thm:main1}. 
We start with the Koteck{\'y}-Preiss type criterion~\cite{kotecky-preiss1986}, as  extended to soft-core and continuum systems by Ueltschi in~\cite{ueltschi2004} \blue{(and to stable interactions by Ueltschi and Poghosyan in~\cite{poghosyan-ueltschi2009})}. 

\begin{kor}
	\label{KP}
	Let $z$ be a non-negative activity function and assume stable interactions in the sense of \eqref{def:stab} for some $B\geq 0$. If there exists a measurable function $a: \mathbb{X}\to \R_+$ such that for  all	$x\in\mathbb{X}$
	\begin{equation}\label{KPcond}
		\int_\mathbb X \bigl|f(x,y)\bigr| \e^{a(y)} \lambda_z(\dd y)+2B(x)  \leq a(x), 
\end{equation}
then  the activity expansions $\rho_n(x_1,\ldots,x_n;z)$ converge absolutely and the bounds 
$$
\rho_n(x_1,\ldots,x_n;z)\leq z(x_1)\cdots z(x_n) \e^{a(x_1)+\cdots + a(x_n)}
$$
hold for all $n\in\N$ and $(x_1,\ldots, x_n)\in \mathbb X^n$. Notice that for non-negative pair interactions, we can choose $B\equiv 0$ in condition~\eqref{KPcond}.
\end{kor}
\begin{rem}
Notice that via the substitution $\hat{a}=a-2B$ the above criterion is equivalent to the existence of a measurable function $\hat a: \mathbb{X}\to \R_+$ such that for all $x\in\mathbb X$
$$
\int_\mathbb X \bigl|f(x,y)\bigr| \e^{\hat a(y)+2B(y)} \lambda_z(\dd y)  \leq \hat a(x).
$$
\end{rem}
\begin{proof}
	Assume that (\ref{KPcond}) holds and define $\vect  \xi=(\xi_n)_{n\in\mathbb{N}},\ \xi_n:\mathbb{X}^n\to [0, \infty)$, by 
	$$
		\xi_n(x_1,...,x_n):=z(x_1)\cdots z(x_n) \e^{a(x_1)+\cdots + a(x_n)}
	$$ 
	for some $a(\cdot)$ satisfying (\ref{KPcond}). The interactions fulfill the stability condition \eqref{def:stab}, therefore for every $n\in \N$ and $x_1,...,x_n\in\mathbb{X}$ there exists an index $j\in\{1,\ldots,n\}$ such that the bound 
\begin{align}\label{stab_ub}
\prod\limits_{ 1\leq i\leq n,i\neq j} (1+f(x_j,x_i))\leq \e^{2B(x_j)}
\end{align}
holds. Choose the selection rule $s$ that always picks an element $x_j$ satisfying \eqref{stab_ub} from $(x_1,\ldots,x_n)$. Plugging our choice of  $\vect  \xi$ into the left-hand side of Eq.~\eqref{eq:condition1} and bounding the interaction term as $\prod_{i=2}^n(1+f(x_s,x'_i))\leq \e^{2B(x_s)}$, we recognize an exponential series, and find altogether that 
	the left-hand side of~\eqref{eq:condition1} is bounded by 	
	$$
		z(x_s)z(x'_2)\cdots z(x'_n)\, 	\e^{a(x'_2)+\cdots + a(x'_n)} \exp\Bigl( \int_\mathbb X |f(x_s,y)| \e^{a(y)} \lambda_z(\dd y)+2B(x_s) \Bigr). 
	$$
	By condition~\eqref{KPcond}, this is in turn bounded by $\xi_{n}(x_1,\ldots, x_n)$.  It follows that condition (i) of Theorem~\ref{thm:main1} is satisfied. 
\end{proof}

Analogously, one shows that the criterion by Fern{\'a}ndez and Procacci~\cite{fernandez-procacci2007}, extended to soft-core and continuum systems \blue{by Faris in~\cite{faris2008} and }by Jansen in~\cite{jansen2019clustergibbs},  is sufficient for absolute convergence of the activity expansions $\vect \rho$. We prove the result in the slightly more general setup of locally stable interactions.

\begin{kor}
\label{FPcrit}
Let $z$ be a non-negative activity function and assume locally stable interactions in the sense of \eqref{def:loc_stab} for some $C\geq 0$. If there exists a measurable function $\mu: \mathbb{X}\to [0,\infty)$ such that for  all $x\in \mathbb X$
\begin{equation}\label{FPcond}
	z(x)\left(1+\sum\limits_{k=1}^{\infty}\frac{1}{k!}\int_{\mathbb{X}^{k}}\e^{\sum_{j=1}^k C(y_j)}\prod\limits_{j=1}^{k}\left|f(x,y_j)\right|	\prod\limits_{1\leq i<j\leq k}(1+f(y_i,y_j))  \lambda_\mu^{k}(\dd\vect y)\right) \leq \mu(x),
\end{equation}
then the activity expansions $\rho_n(x_n,\ldots, x_n;z)$ converge absolutely and the bounds
$$
\rho_n(x_n,\ldots, x_n;z)\leq \prod\limits_{1\leq i<j\leq n}\left(1+f(x_i,x_j)\right)\prod\limits_{i=1}^n\mu(x_i)
$$
hold for all $n\in \N$ and $(x_1,\ldots, x_n)\in \mathbb X^n$. Notice that for non-negative pair interactions, we can choose $C\equiv 0$ in condition~\eqref{FPcond}.
\end{kor}
\begin{rem}
\blue{The Fern{\'a}ndez-Procacci condition improves on the the Koteck{\'y}-Preiss condition --- in the sense that the assumptions of  Corollary~\ref{KP} yield the assumptions of Corollary~\ref{FPcrit}. In other words, Corollary~\ref{FPcrit} in general guarantees convergence of $\vect \rho$ on a larger domain of activities
.}
\end{rem}

\begin{proof}
Assume that (\ref{FPcond}) holds and define $\vect \xi=(\xi_n)_{n\in\mathbb{N}},\ \xi_n:\mathbb{X}^n\to [0, \infty),$ by 
$$
	\xi_n(x_1,\ldots,x_n):=\prod\limits_{1\leq i<j\leq n}\left(1+f(x_i,x_j)\right)\prod\limits_{i=1}^n\mu(x_i)
$$
for some $\mu$ satisfying \eqref{FPcond}.  Let $s$ be the selection rule that always selects the first entry --- so that $x_s = x_1$ and $x'_i = x_i$ for $i\geq 2$. For locally stable pair potentials, we have 
\begin{align}\label{FPbound:locstab}
    &\nonumber\prod\limits_{i=2}^n(1+f(x_1,x_i))\xi_{n+k-1}(x_2,\ldots,x_n,y_1,\ldots,y_k)	\\
	 &\nonumber= \prod\limits_{1\leq i<j\leq n}(1+f(x_i,x_j))\prod\limits_{1\leq i<j\leq k}(1+f(y_i,y_j))\prod\limits_{i=2}^n\prod\limits_{j=1}^k(1+f(x_i,y_j))\prod\limits_{i=2}^n\mu(x_i)\prod\limits_{j=1}^k\mu(y_j)\\
	 &\leq \Bigl( \prod\limits_{1\leq i<j\leq n}(1+f(x_i,x_j)) \prod\limits_{i=2}^n\mu(x_i)\Bigr) 
	 \Bigl( \prod\limits_{1\leq i<j\leq k}(1+f(y_i,y_j)) \prod_{j=1}^k \mu(y_j)\Bigr)\e^{\sum_{j=1}^k C(y_j)}, 
\end{align}
where we used the local stability to estimate
\begin{align*}
\prod\limits_{i=2}^n\prod\limits_{j=1}^k(1+f(x_i,y_j))=\prod\limits_{j=1}^k\prod\limits_{i=2}^n(1+f(x_i,y_j))\leq\prod\limits_{j=1}^k\e^{C(y_j)}=\e^{\sum_{j=1}^k C(y_j)}.
\end{align*}

We plug our choice of $\vect \xi$ into the left-hand side of~\eqref{eq:condition1} and use the estimate \eqref{FPbound:locstab} together with the assumption \eqref{FPcond} to find that condition (i) of Theorem \ref{thm:main1} is satisfied. 
\end{proof}

\begin{rem}\blue{We see that Theorem~\ref{thm:main1} provides a mechanism to prove sufficient conditions for absolute convergence --- by constructing a sequence of ansatz functions $\vect \xi$ tailored  to satisfy the Kirkwood-Salsburg inequalities under the given condition. Conversely, given an appropriate sequence of ansatz functions $\vect \xi$, obtained, for example, as an approximation of $\vect \rho$, one can try to determine the corresponding sufficient condition for convergence.}
\end{rem}

We now proceed to demonstrate the usefulness of that approach by deriving a sufficient condition that improves on the classical examples above.

\subsection{Abstract polymer models}\label{sec:abstract-polymers}
In the following we want to consider the setup of abstract polymers~\cite{bissacot-fernandez-procacci2010, fernandez-procacci2007}, in which the two classical conditions above --- Koteck{\'y}-Preiss and Fern{\'a}ndez-Procacci --- were first introduced.

Let $\mathbb X$ be a countable set (the set of polymers), let $\mathcal X$ be the powerset of $\mathbb X$ and let $\lambda$ simply be given by the counting measure. Moreover, let $R\subset\mathbb{X}\times\mathbb{X}$ be a symmetric and reflexive relation.  We write $x\nsim y$ for $(x,y)\in R$ (and say that $x$ and $y$ are \emph{incompatible}) and $x\sim y$ for $(x,y)\notin R$ (and say that x and y are \emph{compatible}). Moreover, we call a subset $X\subset\mathbb X$ compatible if $x\sim y$ for all $x\neq y \in X$  and write $X\sim z$ for $z\in \mathbb{X}$ if $z\sim x$ for all $x\in X$. We set $\Gamma{(x)}:=\{y\in\mathbb X\vert\ y\nsim x\}$ for any $x\in\mathbb{X}$ and extend this notation to $\Gamma{(X)}:=\cup_{x\in X}\{y\in\mathbb X\vert\ y\nsim x\}$ for any $X\subset\mathbb{X}$. Notice that we do not require $\Gamma(x)$ to be finite sets and that $x\in\Gamma(x)$ for every $x\in\mathbb{X}$. Finally, we consider hard-core interactions corresponding to Mayer's $f$ function given by $f(x,y):= -\1_{\{x\nsim  y\}}$.

In this setting we prove a new, improved sufficient condition for absolute convergence of the activity expansions $\vect \rho$.
\begin{prop}\label{goodprop}
Let $z$ be a non-negative activity function and assume that there exists $\mu:\mathbb{X}\to[0,\infty)$ such that for all $x\in\mathbb{X}$
\begin{equation}\label{Newcond}
z(x)\left(1+\sum\limits_{k\geq 1}\sum\limits_{{\substack{Y=\{y_1,...,y_k\}\\y_i \nsim x,\ y_i\sim y_j}}}\prod\mu(y_i)\prod\limits_{w\in\Gamma(Y)}\e^{\mu(w)}\right)\leq\mu(x)\prod\limits_{w\in\Gamma(x)}\e^{\mu(w)},
\end{equation}
where the inner sum on the left-hand side runs over compatible subsets $Y=\{y_1,...,y_k\}\subset \Gamma(x)$. Then the activity expansions $\rho_n(x_1,\ldots, x_n;z)$ converge absolutely and the bounds
$$
\rho_n(x_1,\ldots, x_n;z)\leq\prod\limits_{1\leq i<j\leq n}\mathbbm{1}_{\{x_i\sim x_j\}}\prod\limits_{i=1}^n\mu(x_i)\prod\limits_{w\in\Gamma(\{x_1,\ldots ,x_n\})}\e^{\mu(w)}
$$
hold for all $n\in \N$  and all $(x_1,\ldots, x_n)\in \mathbb X^n$.
\end{prop}
The proof of the proposition essentially exploits the following auxiliary result:\begin{lemma}\label{kl2}
Let $\mu: \mathbb{X}\to[0,\infty)$. Then the following holds for every $x_1\in\mathbb{X}$, $n\in\mathbb{N}$ and $X=\{x_2,...,x_n\}\subset \mathbb{X}$ such that $x_1\sim x_i$ for all $i\in \{2,\ldots,n\}$:
\begin{align}\label{eqkl2}
\frac{\mu(x_1)\prod\limits_{w\in\Gamma(x_1)}\e^{\mu(w)}}{1+\sum\limits_{k\geq 1}\sum\limits_{{\substack{Y=\{y_1,...,y_k\}\\y_i \nsim x_1,\ y_i\sim y_j}}}\prod\limits_{i=1}^k\mu(y_i)\prod\limits_{w\in\Gamma(Y)}\e^{\mu(w)}}\leq \frac{\mu(x_1)\prod\limits_{w\in\Gamma(x_1)\cap\Gamma(X)^C}\e^{\mu(w)}}{1+\sum\limits_{k\geq 1}\sum\limits_{\substack{\substack{Y=\{y_1,...,y_k\}\\y_i \nsim x_1,\ y_i\sim y_j}\\y_i\sim X}}\prod\limits_{i=1}^k\mu(y_i)\prod\limits_{w\in\Gamma(Y)\cap\Gamma(X)^C}\e^{\mu(w)}},
\end{align}
where $\Gamma(W)$ is given by $\cup_{i=1}^{n}\Gamma(w_i)$ for any $n\in\mathbb{N}$ and $W=\{w_1,...,w_n\}\subset\mathbb{X}$. The inner sum in the denominator on the left-hand side runs over compatible subsets $Y=\{y_1,\ldots,y_k\}\subset\Gamma(x_1)$; the inner sum in the denominator on the right-hand side runs over all such subsets $Y$ which additionally satisfy the constraint $Y\cap \Gamma(X)=\varnothing$, i.e., $y_i\sim X$ for all $i\in\{1,\ldots,k\}.$
\end{lemma}
The lemma is of rather technical nature; for the interested reader, the proof is to be found in Appendix~\ref{AppB}.
\begin{rem}\blue{The general idea behind the proof of Proposition~\ref{goodprop} is to argue as in the proofs of the classical  conditions presented in the previous section (Corollary~\ref{KP} and Corollary~\ref{FPcrit}) --- but to choose a sequence of ansatz functions $ \vect\xi$ which, heuristically speaking, encode more of the structure of the exact solution to the Kirkwood-Salsburg equations (i.e., of the activity expansions $\vect \rho$) than the ansatz functions chosen in the proof of those corollaries. The intuition thereby is that ``less multiplicative" ansatz functions $\vect \xi$ provide better convergence criteria.}
\end{rem}

\begin{proof}[Proof of Proposition \ref{goodprop}]
Assume that \eqref{Newcond} holds and define $\vect \xi=(\xi_n)_{n\in\mathbb{N}},\ \xi_n:\mathbb{X}^n\to [0, \infty),$ by setting
\begin{align}
\xi_n(x_1,...,x_n):=\prod\limits_{1\leq i<j\leq n}\mathbbm{1}_{\{x_i\sim x_j\}}\prod\limits_{i=1}^n\mu(x_i)\prod\limits_{w\in\Gamma(\{x_1,\ldots ,x_n\})}\e^{\mu(w)}
\end{align}
for some $\mu$ satisfying \eqref{Newcond}, for any $n\in\mathbb{N}$ and every $(x_1,..,x_n)\in\mathbb{X}^n$. Thereby we again use the convention $\Gamma(\{w_1,\ldots,w_n\})=\cup_{i=1}^n\Gamma(w_i)$ for $\{w_1,...,w_n\}\subset\mathbb X$. As in the preceeding proofs of the classical sufficient conditions, we show that our choice of $\vect \xi=(\xi_n)_{n\in\mathbb{N}}$ satisfies the system of Kirkwood-Salsburg inequalities \eqref{eq:condition1} from Theorem \ref{thm:main1}. To lighten the notation, we choose the same selection rule $s$ as in the proof of Lemma \ref{FPcrit} and  denote by $X$ the set $\{x_2,...,x_n\}$. Notice that the left-hand side of \eqref{eq:condition1} is equal to

\begin{align*}
&z(x_1)\prod\limits_{1\leq i<j\leq n}\mathbbm{1}_{\{x_i\sim x_j\}}\prod\limits_{i=2}^n\mu(x_i)\prod\limits_{w\in\Gamma(X)}\e^{\mu(w)}\\
&\times \left(1+\sum\limits_{k\geq 1}\sum\limits_{Y=\{y_1,...,y_k\}}\prod\limits_{j=1}^{k}\mathbbm{1}_{\{y_j\nsim x_1\}}\prod\limits_{\substack{2\leq i \leq n\\1\leq j\leq k}}\mathbbm{1}_{\{x_i\sim y_j\}}\prod\limits_{1\leq i< j\leq k}\mathbbm{1}_{\{y_i\sim y_j\}}\prod\limits_{j=1}^k\mu(y_j)\prod\limits_{w\in\Gamma(Y)\cap\Gamma(X)^C}\e^{\mu(w)}\right).
\end{align*}

By Lemma $\ref{kl2}$, the assumption that $z$ satisfies the condition \eqref{Newcond} implies that $z$ also satisfies the inequality
\begin{align*}
&z(x_1)\left(1+\sum\limits_{k\geq 1}\sum\limits_{Y=\{y_1,...,y_k\}}\prod\limits_{j=1}^{k}\mathbbm{1}_{\{y_j\nsim x_1\}}\prod\limits_{\substack{2\leq i \leq n\\1\leq j\leq k}}\mathbbm{1}_{\{x_i\sim y_j\}}\prod\limits_{1\leq i< j\leq k}\mathbbm{1}_{\{y_i\sim y_j\}}\prod\limits_{j=1}^k\mu(y_j)\prod\limits_{w\in\Gamma(Y)\cap\Gamma(X)^C}\e^{\mu(w)}\right)\\&\leq
\mu(x_1)\prod\limits_{w\in\Gamma(x_1)\cap\Gamma(X)^C}\e^{\mu(w)}
\end{align*}
and thus, for our choice of $\vect\xi$, the left-hand side of \eqref{eq:condition1} is bounded from above by
\begin{align*}
&\prod\limits_{1\leq i<j\leq n}\mathbbm{1}_{\{x_i\sim x_j\}}\prod\limits_{i=1}^n \mu(x_i)\prod\limits_{w\in\Gamma(X)}\e^{\mu(w)}\prod\limits_{w\in\Gamma(x_1)\cap\Gamma(X)^C}\e^{\mu(w)}\\&=\prod\limits_{1\leq i<j\leq n}\mathbbm{1}_{\{x_i\sim x_j\}}\prod\limits_{i=1}^n\mu(x_i)\prod\limits_{w\in\Gamma(X\cup\{x_1\})}\e^{\mu(w)}=\xi_{n}(x_1,...,x_n),
\end{align*}
which --- by Theorem \ref{thm:main1} --- yields the claim of the proposition.\end{proof}

\begin{ex}
Consider non-overlapping (hard-core interactions) cubes on $\mathbb{Z}^2$ of side-length $2$ with translationally invariant activity $z$. The sufficient condition on $z$ for the absolute convergence of $\rho(z)$ given by the Fern{\'a}ndez-Procacci criterion provides the bound
\begin{align*}
z\leq \max\limits_{\mu\geq 0}\ \ \frac{\mu}{1+9\mu+16\mu^2+8\mu^3+\mu^4}\approx 0.057271,
\end{align*}
while our condition from Proposition \ref{goodprop} provides 
\begin{align*}
z\leq \max\limits_{\mu\geq 0}\ \ \frac{\mu e^{9\mu}}{1+9e^{9\mu}\mu+(6e^{15\mu}+8e^{16\mu}+2e^{17\mu})\mu^2+8e^{21\mu}\mu^3+e^{25\mu}\mu^4}\approx 0.060833.
\end{align*}
This corresponds to an improvement of approximately 6 percent.
\end{ex}

\subsection{Hard-core systems in the continuum} \label{sec:hardcore}

Let $\mathscr K'$ be the collection of non-empty compact subsets of $\R^d$, equipped with the Hausdorff distance and Borel $\sigma$-algebra~\cite[Chapter I-4]{matheron1975book}, and $\mathbb X\subset \mathscr K'$ a non-empty measurable subset. \blue{Here we want to additionally assume that $\mathbb X$ consists of bounded convex sets that are non-empty and regular closed, i.e., that are equal to the closure of its non-empty interior. Notice that such sets are compact and have finite positive Lebesgue measure that is equal to the Lebesgue measure of their interior.} In practice $\mathbb X$ will consist of easily described subsets. For example, when dealing with closed balls $B_r(x)\subset \R^d$ we may identify $\mathbb X$ with $\R^d\times \R_+$. Consider the hard-core interactions given by the potential $v(X,Y) := \infty \1_{\{X\cap Y \neq \varnothing \}}$, Mayer's $f$ function then is 
$$
	f(X,Y) = - \1_{\{X\cap Y\neq \varnothing\}}.
$$
Clearly the function is well-defined for general subsets $X,Y\subset \R^d$ that are not necessarily in $\mathbb X$, the domains of definition of the functions $\varphi_n^\mathsf T$ and  $\psi_{n,n+k}$ extend accordingly. 

For $D\subset \R^d$ and a measure $\lambda_z$ on $\mathcal X$ defined as in \eqref{def:lambda_z}, consider the formal series
\be \label{eq:TDdef}
	T(D;z) := 1+\sum_{k=1}^\infty \frac{1}{k!}\int_{\mathbb X^k} \varphi_{1+k}^\mathsf T(D,Y_1,\ldots, Y_k) \lambda_z^k(\dd \vect Y).
\ee
As is well-known\blue{~\cite[Eq. (3.12)]{faris2010species}}
\be \label{eq:tdz}
	T(D;z) = \exp\Biggl( -  \sum_{k=1}^\infty \frac{1}{k!}\int_{\mathbb X^k} \1_{\{\exists i:\, Y_i \cap D\neq \varnothing\}} \varphi_{k}^\mathsf T(Y_1,\ldots, Y_k) \lambda_z^k(\dd \vect Y) \Biggr) 
\ee
on the level of formal power series. 

Moreover, if the domain $D$ can be written as a finite union of disjoint objects $X_i\in\mathbb{X}$, say $D=X_1\cup\ldots\cup X_n$ for $n\in\mathbb{N}$, then the identity
\begin{align*}
1+\sum_{k=1}^\infty \frac{1}{k!}\int_{\mathbb X^k} \varphi_{1+k}^\mathsf T(D,Y_1,\ldots, Y_k) \lambda_z^k(\dd \vect Y)=\sum_{k=0}^\infty \frac{1}{k!} \int_{\mathbb X^k} \psi_{n,n+k}(X_1,\ldots,X_n,Y_1,\ldots,Y_k)\lambda_z^k(\dd \vect Y)
\end{align*}
holds by Lemma~\ref{lem:hc1} below and we recognize that the series $T(D,z)$ provide expansions for the reduced correlation functions in the sense that 
$$
	\rho_n(X_1,\ldots,X_n;z) = z(X_1)\cdots z(X_n) \1_{\{X_1,\ldots,X_n\ \text{disjoint}\}}\,  T(X_1\cup\cdots \cup X_n;z). 
$$
\blue{The absolute convergence of the expansions $\rho(z)$ for the correlation functions is implied by the absolute convergence of $T(D;z)$, i.e., by the pointwise convergence 
$$
1+\sum_{k=1}^\infty \frac{1}{k!}\int_{\mathbb X^k} \vert \varphi_{1+k}^\mathsf T(D,Y_1,\ldots, Y_k)\vert \lambda_z^k(\dd \vect Y)<\infty,
$$
for all domains $D$ that are unions of finitely many  objects $X_i \in \mathbb X$.} 

Assume we are given a systematic way to chop up the objects $X\in \mathbb X$ into smaller bits and pieces, called \emph{snippets} (\red{think: analogous to representing a polymer as a collection of monomers in the discrete setup of subset polymers}). That is, choose a positive number $\varepsilon>0$ and assume that
there is a designated collection  $\mathbb E_\eps$ of bounded Borel sets in $\R^d$, each of which is contained in some open ball of radius $\varepsilon$, and a \textit{chopping map}
$$
	C:\mathbb X\to \mathcal P(\mathbb E_\eps),\quad X\mapsto C(X) 
$$
such that for every $X\in \mathbb X$, $C(X) = \{E_1,\ldots, E_m\}$ with $m\in \N$ and $E_1,\ldots, E_m$ a set partition of $X$. We additionally want to assume that the topological boundary of every snippet  is a $\lambda$-null set, i.e., $\lambda(\overline{E}\backslash E^\circ)=0$ for all $E\in\mathbb E_\eps$ (where $\overline{E}$ denotes the topological closure and $E^\circ$ the interior of $E$).

Let $\mathbb D_\eps$ be the set of bounded domains $D\subset \R^d$ that can be written as the union of finitely many disjoint snippets. The empty set $D= \varnothing$ is an element of $\mathbb D_\eps$. For two disjoint subsets $D_0,D_1\subset\R^d$ with $D_0\neq \varnothing$ and for finitely many objects $Y_1,\ldots,Y_k\in\mathbb{X}$, $k\in\mathbb N$, set 
\be \label{eq:idef}
	I(D_0; D_1; Y_1,\ldots, Y_k):=\Bigl( \prod_{i=1}^k \1_{\{D_0\cap Y_i \neq \varnothing,\, D_1\cap Y_i = \varnothing\}}\Bigr)\Bigl( \prod_{1\leq i<j\leq k} \1_{\{Y_i\cap Y_j = \varnothing\}}\Bigr).
\ee

\begin{thm} \label{thm:main2}
	 Let $z(\cdot)$ be a non-negative activity function. 
	The following two conditions are equivalent: 
	\begin{enumerate} 
		 
		\item [(i)] There exists a non-negative map $a:\mathbb D_\eps \to \R_+$ such that for all $D\in \mathbb D_\eps$, the \red{map $\mathscr K'\ni F\mapsto a(D\cup F)$ is measurable} and the following system of inequalities is satisfied: For all non-empty $D\in \mathbb D_\eps$ with $C(D)=\{E_1,\ldots,E_m\}\subset\E_\eps$ for some $m\in\mathbb{N}$, there exists an $s\in\{1,\ldots,m\}$ such that --- setting $D':=D\backslash E_s$ --- we have
\begin{align*}
			\sum_{k=1}^\infty \frac{1}{k!}\int_{\mathbb X^k} I(E_s; D'; Y_1,\ldots, Y_k) \e^{a(D'\cup Y_1\cup \cdots \cup Y_k) - a(D')} \lambda_z^k(\dd \vect Y) 
			\leq \e^{a(E_s\cup D') - a(D')}-1. 
		\end{align*}
	\item [(ii)] $T(D;z)$ is absolutely convergent for all $D\in \mathbb D_\eps$.
	
	\end{enumerate} 
	Moreover, if one of the equivalent conditions (hence, both) holds true, then, for all $D\in \mathbb D_\eps$, we have
	\be \label{eq:abound}
		\bigl| \log T(D;z)\bigr|\leq \sum_{k=1}^\infty \frac{1}{k!}\int_{\mathbb X^k} \1_{\{\exists i:\, Y_i \cap D\neq \varnothing\}} \bigl| \varphi_{k}^\mathsf T(Y_1,\ldots, Y_k) \bigr| \lambda_z^k(\dd \vect Y) \leq a(D). 
	\ee
\end{thm}

\subsection{Subset polymers}\label{sec:sub_poly}
Let $\mathbb X$ consist of the finite non-empty subsets of $\Z^d$ (or any other countable set), and let $\mathcal X = \mathcal P(\mathbb X)$ be the $\sigma$-algebra containing all subsets of $\mathbb X$. The reference measure $\lambda$ is simply the counting measure. 
The interaction is a pure hard-core interaction as in Section~\ref{sec:hardcore}. Notice that this setup is a special case of the abstract polymer setup introduced in Section~\ref{sec:abstract-polymers}.
For a finite set $D\subset \Z^d$, define $T(D;z)$ as in~\eqref{eq:TDdef}. 
In statistical physics $T(D;z)$ corresponds to the probability that no polymer intersects $D$. If $D$ is a polymer or a union of disjoint polymers, it corresponds to a \emph{reduced correlation function} in the sense of~\cite{gruber-kunz1971}.

Notice how in the case of subset polymers every polymer always can be ``chopped" in a canonical way --- into a disjoint collection of monomers. Those play the role of snippets from the previous section --- that simplifies the formulation of a criterion for absolute convergence of the activity expansions $\vect \rho$ (compare next result with Theorem \ref{thm:main2}).

\begin{thm}
	\label{thm:bfp}
	Let $(z(X))_{X\in \mathbb X}$ be a non-negative activity. 
	The following two conditions are equivalent:
	\begin{enumerate}
		\item [(i)] There exists a function $a(\cdot)$ from the finite subsets of $\Z^d$ to $[0,\infty)$ such that $a(\varnothing )=0$ and the following system of inequalities is satisfied: For all finite, non-empty subsets $D\subset\mathbb{Z}^d$ there exists an $x\in D$ such that --- setting $D':=D\backslash\{x\}$ --- we have
		\be \label{eq:discrete-nc}
			\sum_{\substack{Y\in \mathbb X:\\ Y\ni x,\, Y\cap D' =\varnothing}} z(Y) \e^{a(D'\cup Y) - a(D')} \leq \e^{a(D'\cup \{x\}) - a(D')}-1.
		\ee
		
		\item [(ii)] $T(D;z)$ is absolutely convergent for all finite subsets $D\subset \Z^d$. 
		
	\end{enumerate} 
	Moreover, if one of the equivalent conditions (hence, both) holds true, then, for all finite subsets $D\subset \Z^d$, we have
	\begin{align}\label{eq:abound_disc}
		\bigl| \log T(D;z)\bigr|\leq \sum_{k=1}^\infty \frac{1}{k!}\sum_{(Y_1,\ldots, Y_k) \in \mathbb X^k}  \1_{\{\exists i:\, Y_i \cap D\neq \varnothing\}} \bigl| \varphi_{k}^\mathsf T(Y_1,\ldots, Y_k) \bigr| z(Y_1)\cdots z(Y_k) \leq a(D). 
	\end{align}
\end{thm}

\noindent The theorem is similar to Claim~1 in~\cite[Section 4.2]{bissacot-fernandez-procacci2010}. As noted in~\cite{bissacot-fernandez-procacci2010}, Theorem~\ref{thm:bfp} allows for an easy recovery of the \emph{extended Gruber-Kunz criterion}. The criterion is named after Gruber and Kunz~\cite{gruber-kunz1971}, who proved a similar condition but with a strict inequality. See~\cite{fernandez-procacci2007} for a comparison of the Gruber-Kunz criterion to other classical conditions.

\begin{kor} \label{cor:GK}
Let $(z(X))_{X\in \mathbb X}$ be a non-negative activity. Suppose there exists some $\alpha \geq 0$ such that for all $x\in \Z^d$, 
\be\label{gens}
	\sum_{Y\ni x}z(Y)\,\e^{\alpha\mid Y|}\leq \e^{\alpha}-1.
\ee
Then $T(D;z)$ is absolutely convergent, for all finite subsets $D\subset \Z^d$. 
\end{kor}

\begin{proof}
Set $a(D):=\alpha|D|$, where $\alpha>0$ satisfies the inequality from \eqref{gens}. Because of the additivity of $a(\cdot)$, we have $a(D\cup Y) = a(D) + a(Y)$ for all finite, disjoint subsets $D,Y\subset \Z^d$. Therefore condition~\eqref{eq:discrete-nc} becomes 
	$$
		\sum_{\substack{Y\ni x:\\ Y\cap D =\varnothing}}z(Y)\, \e^{a(Y)} \leq \e^{a(\{x\})}-1,
	$$
which depends on $D$ only through the constraint $Y\cap D\neq\varnothing$ on the left-hand side. By the non-negativity of the activity $z$, it is clearly sufficient that
	$$
		\sum_{Y\ni x}z(Y)\, \e^{a(Y)} \leq \e^{a(\{x\})}-1,
	$$
	which holds true for all $x\in\mathbb{Z}^d$ because of~\eqref{gens}.
\end{proof}

Another  immediate consequence of Theorem~\ref{thm:bfp} is that convergence of cluster expansions implies exponential decay of the activities in the object size. Precisely, set 
$$
	V(D):= \sum_{\substack {Y\in \mathbb X:\\Y \cap D \neq \varnothing}} z(Y).
$$
Notice that if the activity is translationally invariant and not identically zero, one can choose an arbitrary polymer $X\in\mathbb{X}$ with positive activity, say $z_0> 0$, and obtain the bound
\begin{align}\label{poly_rateofdecay}
	V(D) \geq z_0 |D|.
\end{align}

\begin{thm} \label{thm:subset-necessary} 
	If $z(\cdot)$ is a non-negative activity and $T(D;z)$ is absolutely convergent for all finite subsets $D\subset \Z^d$, then necessarily 
	$$
		\sum_{Y \ni x} z(Y) \e^{V(Y)}  <\infty
	$$
	for all $x\in \mathbb{Z}^d$. 
\end{thm} 

\begin{proof}
	By condition (i) in Theorem~\ref{thm:bfp}, evaluated at $D'=\varnothing$, there exists a non-negative function $a(\cdot)$ such that 
	$$
		\sum_{Y \ni x} z(Y) \e^{a(Y)} \leq \e^{a(\{x\})} - 1< \infty.
	$$
	For any polymer $Y\in\mathbb{X}$, the value $a(Y)$ is necessarily larger than $V(Y)$  by \eqref{eq:abound_disc} and the claim follows. 
\end{proof} 

\noindent For translationally invariant systems, Theorem~\ref{thm:subset-necessary} says that if the activity expansions are absolutely convergent, then necessarily the activities are exponentially small in the size of the object --- by \eqref{poly_rateofdecay} we can observe that $z(X)  = O(\exp( - \blue{z_0} |X|))$ when $|X|\to \infty$. 
Let us emphasize that the necessary exponential decay is an intrinsic limitation of the activity expansion, which cannot be eliminated by tinkering with different sufficient convergence conditions. Rigorous results for one-dimensional and hierarchical models~\cite{jansen2015tonks, jansen2019hierarchical} suggest that the exponential decay is not needed for the convergence of the multi-species virial expansion, however for general systems this is so far an unproven conjecture.

\section{Combinatorial lemmas. Proof of Theorems~\ref{thm:main1}, \ref{thm:main2}, and~\ref{thm:bfp}}

\subsection{Forest partition schemes. Alternating sign property}

To obtain a better understanding of the series $\rho_n$ given by the generating functions of multi-rooted graphs, we now consider a particular way to construct the latter --- by taking a designated spanning forest and successively adding edges to it. This perspective onto multi-rooted graphs leads to a forest-graph equality analogous to the familiar tree-graph identity for connected graphs~\blue{\cite[Proposition 5]{fernandez-procacci2007}} and allows for a direct proof of an alternating sign property for the coefficients $\psi_{n,n+k}$ of $\rho_n$ in the case of repulsive interactions (i.e., for non-negative potentials). 

The forest-graph equality builds on the notion of  \textit{forest partition schemes} --- maps that assign spanning forests to multi-rooted graphs in $\mathcal D_{n,n+k}$ and thereby in a specific manner provide partitions of $\mathcal D_{n,n+k}$. 

In the following, we let $\mathcal F_{n,n+k}$ denote the set of forest graphs on the vertex set $[n+k]$ consisting of $n$ rooted trees, where the vertices $\{1,\ldots,n\}$ are the roots of the trees (recall that a forest is an acyclic graph and a tree is a connected acyclic graph).

\begin{definition}[Forest partition scheme] \label{def:forestpartition} 
	A \emph{forest partition scheme} is a family of maps $\pi_{n,k}:\ \mathcal D_{n,n+k}\to \mathcal F_{n,n+k}$ such that for all $n\in \N$, $k\in \N_0$, and all $F\in \mathcal F_{n,n+k}$, there exists a graph $R_{n,k}(F) \in \mathcal D_{n,n+k}$ with 
	$$
		\pi_{n,k}^{-1}\bigl( \{F\}\bigr) = \{G\in \mathcal D_{n,n+k} \mid E(F)\subset E(G)\subset E\bigl(R_{n,k}(F)\bigr)\}  =: [F, R_{n,k}(F)]. 
	$$ 
\end{definition} 

To lighten the notation, we introduced partition schemes as families of maps on uncoloured structures. Notice, however, that partition schemes may be defined on coloured structures and may be allowed to depend on the colouring of the vertex set. Therefore, one could introduce families of maps $\pi_{k,n}({\vect x_{[n]}})$, indexed additionally by colourings $\vect x_{[n]}\in\mathbb X^n$ of $[n]=\{1,...,n\}$. Same graphs on the vertex set $[n]$ with different colourings $\vect x_{[n]}$ of the vertices can be mapped onto different forests under such partition schemes.

The existence of forest partition schemes is ensured by the existence of a large class of tree partition schemes, e.g, the Penrose tree partition scheme (\blue{see~\cite{fernandez-procacci2007, ueltschi2017}; for coulouring-dependent schemes see also~\cite{temmel2014,procacci-yuhjtman2017}}).

\begin{ex}	
 A particular forest partition scheme can be defined as follows: For a given multi-rooted graph, construct a connected graph from it by adding a ghost-vertex and connecting it to every root directly by an edge. Then apply the Penrose partition scheme to the resulting connected graph to obtain a spanning tree of this connected graph. Finally, by removing the ghost vertex as well as every edge incident to it, one gets a spanning forest of the initial multi-rooted graph. For the map given by this construction, the characterizing properties of a forest partition scheme follow from the corresponding properties of the Penrose tree partition scheme. 
\end{ex}

Naturally, the choice of the Penrose tree partition scheme in the example above is somewhat arbitrary; any tree partition scheme which does not ``delete" any edge incident to the ghost vertex in the above construction yields a forest partition scheme via the same procedure.

\begin{prop}[Forest-graph equality]\label{forep}
	Let $(\pi_{n,k})_{n\in \N,\, k\in \N_0}$ be a forest partition scheme and \blue{let $(R_{n,k})_{n\in \N,\, k\in \N_0}$ provide the corresponding family of multi-rooted graphs as in Definition~\ref{def:forestpartition}}. 
	Then 
	$$
				\psi_{n,n+k}(x_1,\ldots,x_{n+k}) = \sum_{F\in \mathcal F_{n,n+k}}  \prod_{\{i,j\}\in E(F)} f(x_i,x_j) \prod_{\{i,j\}\in E(R_{n,k}(F))\setminus E(F)} \bigl( 1+ f(x_i,x_j)\bigr)
	$$
	for all $n\in \N$, $k\in \N_0$, and $(x_1,\ldots, x_{n+k})\in \mathbb X^{n+k}$. 
\end{prop} 

\begin{proof} 
	The proof is similar to the standard proof of the tree-graph equality~\blue{\cite[Proposition 5]{fernandez-procacci2007}}. We have 
	\begin{align*} 
		\psi_{n,n+k}(x_1,\ldots,x_{n+k})  &= \sum_{G\in \mathcal D_{n,n+k}} w\bigl(G,(x_1,\ldots,x_{n+k})\bigr) \\
			& = \sum_{F\in \mathcal F_{n,n+k}} \sum_{\substack{G\in \mathcal D_{n,n+k}:\\ \pi_{n,k}(G) = F}} w\bigl(G,(x_1,\ldots,x_{n+k})\bigr) \\
			& = \sum_{F\in \mathcal F_{n,n+k}}  \prod_{\{i,j\}\in E(F)} f(x_i,x_j) \prod_{\{i,j\}\in E(R_{n,k}(F))\setminus E(F)} \bigl( 1+ f(x_i,x_j)\bigr).  \qedhere
	\end{align*} 	
\end{proof} 

\noindent In the case of repulsive interactions, the forest-graph equality allows for a direct proof of the alternating sign property for the graph weights $\psi_{n,n+k}$. For $n=1$, it reduces to the well-known alternating sign property~\blue{\cite[Eq. (2.8)]{fernandez-procacci2007}} 
 \be \label{eq:altsign} 
 	\varphi_n^\mathsf T(x_1,\ldots, x_n) = (-1)^n \bigl| \varphi_n^\mathsf T(x_1,\ldots, x_n)\bigr|
 \ee
of the Ursell functions.

\begin{kor}\label{prop:alt-sign}
	For non-negative potentials, we have
	$$
		\psi_{n,n+k}(x_1,\ldots,x_{n+k}) = (-1)^{k} \bigl|\psi_{n,n+k}(x_1,\ldots,x_{n+k})\bigr|
	$$
	for all $n\in \N$, all $k\in \N_0$, and all $(x_1,\ldots,x_{n+k})\in \mathbb X^{n+k}$. 
\end{kor}
\begin{proof}
		Each forest $F\in \mathcal F_{n,n+k}$ has exactly $k$ edges. Indeed, the forest $F$ consists of trees $T_1,\ldots, T_n$. Let $m_i$ be the number of vertices of the tree $T_i$; thus $m_1+\cdots + m_n =n+k$.  Each tree $T_i$ has exactly $m_i -1$ edges, therefore the number of edges of the forest is given by $\sum_{i=1}^n (m_i -1) = k$. Since $f\leq 0$ and $1+f\geq 0$ for non-negative potentials, it follows that 
	$$
		\psi_{n,n+k}(x_1,\ldots,x_{n+k}) = (-1)^{k}  \sum_{F\in \mathcal F_{n,n+k}}  \prod_{\{i,j\}\in E(F)} |f(x_i,x_j)| \prod_{\{i,j\}\in E(R_{n,k}(F))\setminus E(F)} \bigl( 1+ f(x_i,x_j)\bigr),
	$$
	hence $(-1)^{k}\psi_{n,n+k}(x_1,\ldots, x_{n+k})\geq 0$. 
\end{proof} 

We will use the alternating sign property to establish that --- in the case of non-negative potentials --- condition (i) in Theorem~\ref{thm:main1} is not only sufficient but also necessary for absolute convergence of $ \rho$.

We conclude this section with a lemma that is not needed for the proof of Theorem~\ref{thm:main1} but enters the analysis  of hard-core models, see the proof of Lemma~\ref{lem:hc1} below. 

\begin{lemma}  \label{lem:keylem1}
	For all $n\in \N$,  all $k\in \N_0$, and  all $(x_1,\ldots, x_{n+k})\in \mathbb X^{n+k}$, we have 
	\begin{multline} \label{eq:psipart}
		\psi_{n,n+k}(x_1,\ldots,x_{n+k}) = \prod_{1\leq i < j \leq n}\bigl( 1+ f(x_i,x_j)\bigr) \\
			\times  \sum_{\{V_1,\ldots, V_r\}} \prod_{\ell=1}^r \Biggl( \prod_{\substack{1 \leq i \leq n,\\ j\in V_\ell}} \bigl(1+f (x_i,x_j)\bigr) -1 \Biggr) \varphi_{|V_\ell|}^\mathsf T\bigl((x_j)_{j\in V_\ell}\bigr),
	\end{multline} 
	where the sum runs over all set partitions $\{V_1,\ldots, V_r\}$ of non-root vertices $\{n+1,\ldots, n+k\}$. 
\end{lemma}

\begin{rem}
	The lemma allows for an alternative proof of the alternating sign property from Corollary~\ref{prop:alt-sign}, starting from the well-known alternating sign property of the Ursell function instead of the forest-graph equality. Indeed, the sign of every summand in the right-hand side of~\eqref{eq:psipart} is 
	$$
		(-1)^{r+ \sum_{i=1}^r (|V_i| -1)} = (-1)^{k}.
	$$
\end{rem}

\begin{proof}[Proof of Lemma~\ref{lem:keylem1}]
	For $n=1$, the lemma reduces to a well-known equality for the Ursell functions, see e.g. \cite[Eq. (5.13)]{friedli-velenik2018book}. For $n\geq 2$, the proof is similar, we provide the details for the reader's convenience. 
Every multi-rooted graph $G\in \mathcal D_{n,n+k}$ can be constructed in the following way. On the root set $\{1,\ldots, n\}$ pick an arbitrary graph $G_0$. On the complement of the root set do the following construction: Partition the set of the non-root vertices into $r$ sets $V_1,...,V_r$, $r\leq k$. For every block $V_\ell$, pick a connected graph $G_\ell$ with vertex set $V_\ell$, and in addition a non-empty set of edges 
	$
		E_\ell \subset \bigl\{ \{i,j\}\mid i \in \{1,\ldots, n\}, \, j\in V_\ell\}. 
	$
	Then the graph $G$ with vertices $1,\ldots, n+k$ and edge set given by the union of $E_1,\ldots, E_\ell$ and of  the edge sets of $G_0,G_1,\ldots, G_r$ is in $\mathcal D_{n,n+k}$, its graph weight is 
	$$
		w(G;x_1,\ldots, x_{n+k}) = w(G_0;x_1,\ldots,x_n) \prod_{\ell=1}^r \Biggl (\prod_{\{i,j\}\in E_\ell} f(x_i,x_j)\Biggr) w\bigl(G_\ell;(x_j)_{j\in V_\ell}\bigr).
	$$
	Summation over $G_0$ yields the factor $\prod_{1\leq i < j \leq n} (1+f(x_i,x_j))$.
	Summation over the connected graphs $G_\ell$ yields $\varphi_{|V_\ell|}^\mathsf T(\vect x_{V_\ell})$. 	Finally, summation over the edge sets $E_\ell$ yields the factor $\prod_{1\leq i \leq n,j \in V_\ell}(1+f(x_i,x_j)) - 1$. 
\end{proof}

\subsection{Kirkwood-Salsburg equations. Proof of Theorem~\ref{thm:main1}}\label{subsect:3.2.}

\blue{To prove our main result, Theorem~\ref{thm:main1}, we will show that the activity expansions $\vect \rho$ satisfy the Kirkwood-Salsburg inequalities and, moreover, that equality holds for non-negative interactions. To do so, we will need to establish  a recursive formula for the coefficients $\psi_{n,n+k}$ of $\rho_n$ given in terms of multi-rooted graphs.} 

\begin{lemma} \label{lem:recursion}
	Let $n\in \N$,  $k\in \N_0$ and $(x_1,\ldots, x_n) \in \mathbb X^n$. Abbreviate $s= s(\vect x)$. For $L\subset [k]$, let $\ell$ denote the cardinality of $L$. Then for all $(y_1,\ldots,y_k)\in \mathbb X^k$, 
	\begin{multline*}
			\psi_{n,n+k}(x_1,\ldots,x_n,y_1,\ldots, y_k) =\prod_{\substack{1\leq i \leq n:\\ i\neq s}} \bigl(1+f(x_{s}, x_i)\bigr) \sum_{L\subset [k]} \Bigl(\prod_{i\in L} f(x_ {s}, y_i)\Bigr) \\
				\times\psi_{n-1+\ell, n-1+k}\bigl(x'_2,\ldots, x'_{n}, (y_i)_{i\in L}, (y_j)_{j\in [k]\setminus L} \bigr).
	\end{multline*}
	Furthermore, if $n\geq 2$, 
	$$
		\psi_{n,n}(x_1,\ldots,x_n) =\prod_{\substack{1\leq i \leq n:\\ i\neq s}} \bigl(1+f(x_{s}, x_i)\bigr)\,   \psi_{n-1, n-1}\bigl(x'_2,\ldots, x'_{n} \bigr).
	$$
\end{lemma} 

\noindent The lemma is proven in~\cite[Lemma 4.1]{jansen2019clustergibbs} and holds true as well for pair potentials that may take negative values.  The index set $L$ corresponds to the non-root  vertices adjacent to the selected vertex $s$. \blue{Similar recurrent relation are well-known from the literature and have been employed in the context of both activity and density (virial) expansions for a long time (see, e.g., \cite[Eq. 5]{minlos-poghosyan1977}). }

We now want to translate the recurrence relation for coefficients $\psi_{n,n+k}$ from Lemma \ref{lem:recursion} into integral equations for partial sums and series.
\blue{For a non-negative activity function}, we set
$$
	\tilde \rho_n(x_1,\ldots, x_n;z):= z(x_1)\cdots z(x_n) \sum_{k=0}^{\infty} \frac{1}{k!}\int_{\mathbb X^k} \bigl|\psi_{n,n+k}(x_1,\ldots,x_n,\vect y)\bigr| \lambda_z^k(\dd \vect y).
$$
Notice that $\rho_n(x_1,\ldots,x_n;z)$ is absolutely convergent if and only if $\tilde \rho_n(x_1,\ldots,x_n;z)<\infty$, and, in the case of non-negative potentials, 
$$
	\tilde \rho_n(x_1,\ldots, x_n;z) = (-1)^n \rho_n(x_1,\ldots, x_n;-z)$$
holds due to the alternating-sign property from Corollary~\ref{prop:alt-sign}.

Let  $\vect{\tilde S}_N(z) = (\tilde S_{N,n}(\cdot;z))_{n\in \N}$ be the vector of partial sums given by 
$$
	\tilde S_{N,n}(x_1,\ldots,x_n;z):=z(x_1)\cdots z(x_n) \sum_{k=0}^{N-n} \frac{1}{k!}\int_{\mathbb X^k} \bigl| \psi_{n,n+k}(x_1,\ldots,x_n,\vect y)\bigr| \lambda_z^k(\dd \vect y)
$$
if $N \geq n$, and $0$ otherwise.  The summand for $k=0$  is to be read as $|\psi_{n,n}(x_1,\ldots,x_n)|$.

\begin{prop} \label{tKSmult}
	\red{For general pair-interactions}, we have 
	\begin{align}\label{tKSmult_ineq}
	\vect{\tilde  \rho}(z) \leq \vect e_z+ \tilde K_z^s \vect{\tilde \rho}(z)
	\end{align}
	and 
	$$
		\vect{\tilde S}_1(z) = \vect e_z,\quad \vect{\tilde S}_{N+1} \leq \vect e_z+\tilde K_z^s \vect{\tilde S}_{N}(z)\quad (N\geq 1). 
	$$
	Moreover, for non-negative potentials, we get the equalities 
	\begin{align}\label{tKSmult_eq}
		\vect{\tilde  \rho}(z) = \vect e_z+ \tilde K_z^s \vect{\tilde \rho}(z)
	\end{align}
	and 
	$$
		\quad \vect{\tilde S}_{N+1} = \vect e_z+\tilde K_z^s \vect{\tilde S}_{N}(z)\quad (N\geq 1). 
	$$

\end{prop} 



\begin{proof} 
	The equality $\vect{\tilde S}_1(z) = \vect e_z$ follows from the definition of $\vect {\tilde S}_1(z)$ and $\psi_{1,1}(x_1) =1$. For the recurrence relation, 
	we employ arguments from~\cite[Section~4]{jansen2019clustergibbs} and combine them with the alternating sign property from Corollary \ref{prop:alt-sign} to argue equality in the case of non-negative potentials. 
	
	Consider first $\tilde S_{N+1,n}(z)$ with $2\leq n\leq N+1$. Define
	$$
		\mathscr{R}_{n,\ell} (x_1,\ldots,x_n;y_1,\ldots, y_\ell) := z(x_s) \prod\limits_{i=2}^n(1+f(x_s,x'_i)) \prod_{i=1}^\ell \bigl| f(x_s, y_i)\bigr|.
	$$
	Fix $n\geq 2$. 
	Lemma~\ref{lem:recursion} and the triangle inequality yield 
	\begin{multline}\label{iifps}
		\bigl|\psi_{n,n+k}(x_1,\ldots,x_n,y_1,\ldots, y_k) \bigr|\\
			\leq \sum_{L\subset [k]} \mathscr{R}_{n,\ell} (x_1,\ldots,x_n; \vect y_L)\, \bigl|\psi_{n-1+\ell, n-1+k}\bigl(x'_2,\ldots, x'_{n}, \vect y_L, \vect y_{[k]\setminus L}\bigr) \bigr|,
	\end{multline}
	where $\ell = \#L$. When we integrate over $y_1,\ldots,y_k$, all sets $L$ with the same cardinality contribute the same, therefore 
	\begin{align*}
		& \frac{1}{k!}\int_{\mathbb X^k} \bigl|\psi_{n,n+k}(x_1,\ldots,x_n,y_1,\ldots, y_k) \bigr| \, \lambda_z^k(\dd \vect y) \\
		&\quad \leq\sum_{\ell=0}^k \frac{1}{\ell! (k-\ell)!} \int_{\mathbb X^k} \mathscr{R}_{n,\ell} (x_1,\ldots,x_n; y_1,\ldots,y_\ell)\, \bigl|\psi_{n-1+\ell, n-1+k}\bigl(x'_2,\ldots, x'_{n}, \vect y\bigr) \bigr|\, \lambda_z^k(\dd \vect y). 
	\end{align*} 
	Summing over $k=0,\ldots,N+1-n$  we obtain a double sum over $k$ and $\ell$. A change in summation indices from $(\ell,k)$ to $(\ell,m)=(\ell,k-\ell)$ yields 
	\begin{align*}
		\tilde S_{N+1,n}(x_1,\ldots,x_n;z) & \leq 
			\sum_{\ell=0}^{N+1-n} \frac{1}{\ell!} \int_{\mathbb X^\ell} \mathscr{R}_{n,\ell} (x_1,\ldots,x_n; y_1,\ldots,y_\ell)\Bigl\{\cdots \Bigr\} \lambda_z^\ell\bigl(\dd( y_1,\ldots y_\ell)\bigr),\\
			\{\cdots \} & =  \sum_{m=0}^{N+1-n-\ell} \frac{1}{m!} \int_{\mathbb X^m} \, \bigl|\psi_{n-1+\ell, n-1+\ell+m}\bigl(x'_2,\ldots, x'_{n}, \vect y\bigr) \bigr|\, \lambda_z^m\bigl(\dd (y_{\ell+1},\ldots y_{\ell+m})\bigr).
	\end{align*}
	The term in curly braces is nothing else but $\tilde S_{N,n-1+\ell}(x'_2,\ldots,x'_n,y_1,\ldots,y_\ell)$. For $\ell \geq N+1-n$, the function $\tilde S_{N,n-1+\ell}(\cdot;z)$ is identically zero. It follows that 
	$$
		\tilde S_{N+1,n}(x_1,\ldots,x_n;z) 
			\leq \sum_{\ell=0}^{\infty} \frac{1}{\ell!} \int_{\mathbb X^\ell} \mathscr{R}_{n,\ell} (x_1,\ldots,x_n; \vect y) \tilde S_{N,n-1+\ell}(x'_2,\ldots,x'_n,\vect y)
			 \lambda_z^\ell\bigl(\dd\vect y\bigr).
	$$
	This proves the inequality $\tilde S_{N+1,n}(\cdot;z)\leq(\tilde K_z^s S_{N}(z)\bigr)_n(\cdot)$. The cases $n\geq N+2$ and $n=1$ are treated in a similar fashion, we leave the details to the reader. For the equality  $\tilde S_{N+1,n}(\cdot;z)=(\tilde K_z^s S_{N}(z)\bigr)_n(\cdot)$ in the case of non-negative potentials, notice that for such potentials \eqref{iifps} holds with an equality --- due to the alternating-sign property from Corollary \ref{prop:alt-sign}.
    
	Finally, by passing to the limit $N\to \infty$ in the recurrence relation for $\vect{\tilde S}_N(z)$, the inequality \eqref{tKSmult_ineq} for $\vect{\tilde \rho}(z)$ follows (and in the case of non-negative potentials the fixed point equation \eqref{tKSmult_eq} is obtained). Notice that all exchanges of limits, sums and integrals are permitted by monotone convergence and because all terms involved are non-negative.
\end{proof} 

\begin{proof} [Proof of Theorem~\ref{thm:main1}] \label{proof:main1} For the implication (i) $\Rightarrow$ (ii), suppose there exists a sequence $\vect \xi = (\xi_n)_{n\in \N}$ of measurable non-negative functions $\xi_n:\mathbb X^n\to \R_+$ such that $ \vect e_z + \tilde K_z^s \vect \xi\leq \vect \xi$. We prove by induction over $N$ that $\vect{\tilde S}_N(z) \leq \vect \xi$ for all $N\in \N$. For $N=1$, we have 
	$$
		\vect{\tilde S}_1 = \vect e_z \leq \vect e_z + \tilde K_z^s \vect \xi\leq \vect \xi.
	$$
	If $\vect{\tilde S}_N\leq \vect \xi$ for some $N\in\N$, then 
	$$
		\vect{\tilde S}_{N+1} \leq \vect e_z+\tilde K_z^s \vect{\tilde S}_{N}(z) \leq \vect e_z + \tilde K_z^s \vect \xi\leq \vect \xi,
	$$
where the first inequality holds by Proposition~\ref{tKSmult} and the second one due to the inductive hypothesis and the monotonicity of $\tilde K_z^s$ on non-negative functions.
	 
	This completes the induction and proves $\vect{\tilde S}_N\leq \vect \xi$ for all $N$. Passing to the limit $N\to \infty$, we find $\vect{\tilde \rho} \leq \vect \xi$. This proves the absolute convergence  (ii) as well as the bound~\eqref{rhokbound}. \\
	
	Left to show is the implication \noindent (ii) $\Rightarrow$ (i)  under the additional assumption that the potential is non-negative. Suppose that $\rho_n(x_1,\ldots,x_n;z)$ is absolutely convergent, for all $n\in \N$ and $(x_1,\ldots,x_n)\in \mathbb X^n$. Then $\tilde \rho_n(x_1,\ldots,x_n;z)$ is finite everywhere and we may set 
	$$
		\xi_n(x_1,\ldots,x_n):= \tilde \rho_n(x_1,\ldots,x_n;z).
	$$
	Proposition~\ref{tKSmult} yields $\vect \xi = \vect e_z + \tilde K_z^s \vect \xi$ hence a fortiori $\vect \xi \geq  \vect e_z + \tilde K_z^s \vect \xi$, as a pointwise inequality for all vector entries. This proves (i).\end{proof} 

\subsection{Integral equations for hard-core models. Proof of Theorem~\ref{thm:main2}}\label{sec3.3}

In this section we specialize to hard-core systems in the continuum as in Section~\ref{sec:hardcore} and use capital letters for objects $X\in \mathbb X$. Let $\mathcal D^\mathrm{red}_{n,n+k}\subset \mathcal D_{n,n+k}$ be the collection of graphs $G\in \mathcal D_{n,n+k}$ that have no edges linking any two root vertices  $i,j\in \{1,\ldots,n\}$. Define $\psi_{n,n+k}^\mathrm{red}$ in a similar way as $\psi_{n,n+k}$ but with summation over graphs in $ \mathcal D^\mathrm{red}_{n,n+k}$. It is not difficult to check that 
$$
	\psi_{n,n+k}(X_1,\ldots,X_{n+k}) =\prod_{1\leq i < j \leq n} \bigl(1+ f(X_i,X_j)\bigr)\psi_{n,n+k}^\mathrm{red}(X_1,\ldots,X_{n+k}). 
$$
The reduced \blue{functions} $\psi_{n,n+k}^\mathrm{red}$ satisfy recurrence relations similar to Lemma~\ref{lem:recursion}. 
Define 
\be \label{eq:gdef}
	g(X_1;X_2,\ldots, X_n;Y_1,\ldots, Y_k):= 
\prod_{j=1}^{k}f(X_1,Y_j) \prod_{1\leq i< j\leq k}\bigl(1+f(Y_i,Y_j)\bigr) \prod_{\substack{2\leq i\leq n\\1\leq j\leq k}} \bigl(1+f(X_i,Y_j)\bigr).
\ee
Remember the indicator $I$ from~\eqref{eq:idef} and notice 
\be \label{eq:girel}
	g(X_1;X_2,\ldots, X_n;Y_1,\ldots, Y_k)=(-1)^k I(X_1; X_2\cup \cdots \cup X_n; Y_1,\ldots, Y_k).
\ee

\begin{lemma} \label{lem:redrec} For all $k,n\in \N$, we have
	\begin{multline*}
		\psi_{n,n+k}^\mathrm{red}(X_1,\ldots, X_n,Y_1,\ldots, Y_k) \\
			 = \sum_{L\subset [k]}g\bigl(X_s;X'_2,\ldots, X'_n;(Y_j)_{j\in L}\bigr)
			 \psi_{n-1+\ell,n-1+k}^\mathrm{red}\bigl(X'_2,\ldots, X'_n, (Y_j)_{j\in L}, (Y_j)_{j\in [k]\setminus L}\bigr).
	\end{multline*}
\end{lemma} 

\noindent The proof is based on combinatorial considerations similar to the proof of Lemma~4.1 in~\cite{jansen2019clustergibbs}, we leave the details to the reader. The lemma holds true for arbitrary subsets of $\R^d$, the $X_i$'s and $Y_j$'s need not be in $\mathbb X$. \blue{That applies to our next result as well.}

\begin{lemma}\label{lem:hc1} For all $k,n\in \mathbb{N}$, we have
$$	
	\psi_{n,n+k}^\mathrm{red}(X_1,\ldots,X_n,Y_1,\ldots,Y_{k}) = \varphi_{1+k}^\mathsf T(X_1\cup\cdots \cup X_n,Y_1,\ldots, Y_k).
$$
\end{lemma}

\begin{proof}
	Revisiting the proof of Lemma~\ref{lem:keylem1}, we see that 
	\begin{align} \label{eq:psipartred}
		&\nonumber\psi_{n,n+k}^\mathrm{red}(X_1,\ldots,X_n,Y_1,\ldots,Y_k) \\&= \sum_{\{V_1,\ldots, V_r\}} \prod_{\ell=1}^r \Biggl( \prod_{\substack{1 \leq i \leq n,\\ j\in V_\ell}} \bigl(1+f (X_i,Y_j)\bigr) -1 \Biggr) \varphi_{|V_\ell|}^\mathsf T\bigl((Y_j)_{j\in V_\ell}\bigr),
\end{align}	
	where the sum runs over all set partitions $\{V_1,\ldots, V_r\}$ of non-root vertices $\{n+1,\ldots, n+k\}$. For hard-core interactions, the term in parentheses  is equal to minus the indicator that $D:=X_1\cup \cdots \cup X_n$ is intersected by at least one $Y_j$, $j\in V_\ell$. Thus 
	$$
		\psi_{n,n+k}^\mathrm{red}(X_1,\ldots,X_n,Y_1,\ldots,Y_k) = \sum_{\{V_1,\ldots, V_r\}} \prod_{\ell=1}^r \bigl(  - \1_{\{\exists j\in V_\ell:\, Y_j \cap D\neq\varnothing \}}\bigr) \varphi_{|V_\ell|}^\mathsf T\bigl((Y_j)_{j\in V_\ell}\bigr).
	$$
	Using again~\eqref{eq:psipartred}, we deduce 
	\begin{equation*}
		\psi_{n,n+k}^\mathrm{red}(X_1,\ldots,X_n,Y_1,\ldots,Y_k) = \psi_{1,1+k}(D,Y_1,\ldots, Y_k) = \varphi_{1+k}^\mathsf T (D,Y_1,\ldots,Y_k). \qedhere
	\end{equation*}
\end{proof} 

\begin{lemma} \label{lem:prec}
	Let $k\in\N$ and let $D_0,D_1$ be two disjoint subsets of $\R^d$, with $D_0\neq \varnothing$. Then 
	\begin{multline*}
		\varphi_{1+k}^\mathsf T (D_0\cup D_1,Y_1,\ldots, Y_k) \\
			= \sum_{L\subset [k]}  (-1)^\ell I\bigl( D_0;D_1;(Y_i)_{i\in L}\bigr)\varphi_{1+k-\ell}^\mathsf T \Bigl(D_1\cup \bigl( \bigcup_{i\in L} Y_i\bigr), (Y_j)_{j\in [k]\setminus L}\Bigr),
\end{multline*}
where the sum is taken over all subsets $L\subset[k]$ and $\ell$ denotes the cardinality of $L$. 
\end{lemma} 

\begin{proof}
	The claim of the lemma follows from Eq.~\eqref{eq:girel}, Lemma~\ref{lem:redrec} and Lemma~\ref{lem:hc1}.
\end{proof} 

\noindent For $D\in \mathbb D_\eps$ with $E_1,\ldots,E_n\in \mathbb E_\eps$ and $C(D)=\{E_1,\ldots, E_n\}$, let
$$
 \tilde T(D;z):=1+\sum_{k=1}^\infty \frac{1}{k!}\int_{\mathbb X^k} \bigl|\varphi_{1+k}^\mathsf T(D,Y_1,\ldots,Y_k)\bigr| \lambda_z^k(\dd \vect Y),
$$
$\tilde T_1(D;z):= \delta_{n,1}(\{E_1,\ldots,E_n\})$, and for $N \geq 2$, 
$$
	\tilde T_N\bigl(D;z\bigr)
	:=\mathbbm{1}_{\{n\leq N\}}+\sum_{k=1}^{N-n} \frac{1}{k!}\int_{\mathbb X^k}
		\1_{\{n+\sum_{i=1}^k |C(Y_i)|\leq N\}} \bigl|\varphi_{1+k}^\mathsf T(D,Y_1,\ldots,Y_k)\bigr| \lambda_z^k(\dd \vect Y).
$$
\blue{Although the value of $\tilde T(D;z)$ (if the series converges) does not depend on the choice of the chopping map, notice that the value of $\tilde T_N(D;z)$ clearly does depend on $C(D)=\{E_1,\ldots,E_n\}$ and not only on $E_1\cup\ldots\cup E_n$ --- due to the constraint on the number of snippets.} 

A selection rule is a map $s$ from collections of disjoint snippets $\mathbb D_\eps$ to $\mathbb E_\eps$ such that 
$$
	s(\{E_1,\ldots,E_n\}) \in \{E_1,\ldots, E_n\},
$$
i.e., $s(\cdot)$ selects one of the snippets. We use the suggestive but somewhat abusive notation 
$E_s $ for the selected snippet, and let $E'_2,\ldots, E'_n$ be any enumeration of the remaining snippets.  If $\xi(\cdot)$ is a function from $\mathbb D_\eps$ to $\R_+$ \blue{that satisfies the measurability assumption from Theorem~\ref{thm:main2}}, define a new function $\tilde \kappa_z^s\xi$ \red{(possibly assuming the value ``$\infty$")} by setting 
\begin{align}\label{def:contKS-op}
		\nonumber&(\tilde \kappa_z^s\xi)\bigl( D\bigr)
			:= \1_{\{n\geq 2\}}\, \xi\bigl(E'_2\cup\ldots\cup E'_n) \\
				&+\sum_{k=1}^\infty \frac{1}{k!} \int_{\mathbb X^k} I(E_s;E'_2\cup \cdots \cup E'_n; Y_1,\ldots,Y_k) \xi\Bigl(E'_2\cup\ldots E'_n\cup Y_1\cup\ldots\cup Y_k \Bigr)\lambda_z^k(\dd \vect Y)
\end{align}
for $D\in \mathbb D_\eps$ with $E_1,\ldots,E_n\in \mathbb E_\eps$ and $C(D)=\{E_1,\ldots, E_n\}$.
Furthermore, let $e(D):= \delta_{n,1}(\{E_1,\ldots,E_n\})$ be the indicator that $D$ is a single snippet.

Let $z$ be a non-negative activity such that for every non-empty $D\in \mathbb D_\epsilon$ the series $\tilde T\bigl(D;z\bigr)$ converges absolutely. \blue{\label{arg:m}Notice that the topology induced by the Hausdorff distance is equivalent to the myopic topology and the map $\mathscr K'\ni F\mapsto \1_{\{F\in\mathscr K'\vert\ F\cap B\neq \varnothing\}}(F)=-f(F,B)$ is measurable with respect to the myopic topology for all compact subsets $B$ (see~\cite{molchanov2005}). Measurability of $\mathscr K'\ni F\mapsto \tilde{T}(D\cup F ;z)$ for every $D\in \mathscr K'$ can be concluded, e.g., by representing the series $\tilde T(D\cup F;z)$ as in Equation \eqref{eq:tdz}. Since $\tilde{T}(D\cup F ;z)=\tilde{T}(\overline D\cup F ;z)$ for every $D\in \mathbb D_\eps$, its topological closure $\overline D$ in $\mathbb{R}^d$ and every $F\in \mathscr K'$ (by our assumptions that the boundaries of snippets are $\lambda$-null set), the measurability of $\mathscr K'\ni F\mapsto \tilde{T}(D\cup F ;z)$ for all $D\in \mathbb D_\eps$ follows}.

The next result is an analogue of Proposition~\ref{tKSmult}. 

\begin{prop} \label{prop:cavityint}
 We have 
	$$
		\tilde T\bigl(\cdot;z\bigr) = e(\cdot) +\tilde  \kappa_z^s \tilde T(\cdot;z).
	$$
Moreover $\tilde T_1(\cdot;z) = e(\cdot)$ and for $N\geq 1$
	$$
		\tilde T_{N+1}\bigl(\cdot;z\bigr) = e(\cdot) +  (\tilde{\kappa}_z^s \tilde T_N)\bigl(\cdot;z\bigr).
	$$
\end{prop} 

\noindent The proposition follows from Lemma~\ref{lem:prec} by arguments similar to the proof of Proposition~\ref{tKSmult}, therefore the proof is omitted.

\begin{proof}[Proof of Theorem~\ref{thm:main2}]
		To show the implication (ii) $\Rightarrow$ (i) suppose that $T(D;z)$ is absolutely convergent and thus $\tilde T(D;z)$ is convergent for all non-empty $D\in \mathbb D_\eps$. Moreover, $\tilde T(D;z)$ is uniformly bounded from below by $1$ and does not depend on the choice of the chopping map $C$. We set
		$$
			a(D):= \log \tilde{T}(D;z)\geq 0
		$$
		for non-empty $D\in \mathbb D_\eps$ and $a(\varnothing) := 0$. Furthermore, for $D\in \mathbb D_\eps$ with $E_1,\ldots,E_m\in \mathbb E_\eps$ and $C(D)=\{E_1,\ldots, E_m\}$, let $E_s\in \mathbb E_\eps$ be given by $E_{s(\{E_1,...,E_m\})}$ for some selection rule $s(\cdot)$ and set $D':=D\backslash E_s$. 
	Exploiting the fixed point equation for $\tilde T(\cdot;z)$ from Proposition~\ref{prop:cavityint}, we get 
	\begin{align*}
	\e^{a(D')} + \sum_{k=1}^\infty \frac{1}{k!}\int_{\mathbb X^k}I(E_s;D';Y_1,\ldots, Y_k) \e^{a(D'\cup Y_1\cup \cdots \cup Y_k)} \lambda_z^k(\dd \vect Y) \leq \e^{a(E_s\cup D')}.
	\end{align*}
	Item (i) of Theorem~\ref{thm:main2} follows upon multiplication with $\exp(-a(D'))$ on both sides (\red{in fact, we have shown that the inequality from item (i) holds for every choice of $s\in\{1,\ldots,m\}$}). \\
	
	\noindent  \blue{The implication (i) $\Rightarrow$ (ii) follows from Proposition~\ref{prop:cavityint} by an induction over $N$ similar to the proof of Theorem~\ref{thm:main1} on p.~\pageref{proof:main1}. Indeed, check that for the induction step it is sufficient that the corresponding system of Kirkwood-Salsburg inequalities holds for all finite disjoint unions of snippets (for any choice of the chopping map $C$, the snippet-size $\varepsilon>0$ and the selection rule $s$). Bound \eqref{eq:abound} is then established using the triangle inequality, the alternating sign property, and Eq.~\eqref{eq:tdz}.}
\end{proof}

\subsection{Recurrence relations for subset polymers. Proof of Theorem~\ref{thm:bfp}}
\label{sec:recc_sub}

\blue{Just as in the continuous case, we will show that the expansions $\tilde{T}$ solve a system of integral equations  in the discrete setup. We do so by providing a recursive formula for the corresponding coefficients $\varphi^\mathsf T_{1+k}$.}

\begin{lemma} \label{lem:phirecurr}
	For all finite subsets $D'\subset \Z^d$, all $x\in \Z^d\setminus D'$, and all $k\in \N$, 
	\begin{multline*}
		\varphi_{1+k}^\mathsf T(D'\cup \{x\}, Y_1,\ldots, Y_k) \\
			= \varphi_{1+k}^\mathsf T(D', Y_1,\ldots, Y_k)
			+ \sum_{i=1}^k \bigl( - \1_{\{Y_i\ni x,Y_i \cap D = \varnothing\}}\bigr) \varphi_{k}^\mathsf T \bigl(D'\cup Y_i, (Y_j)_{j\neq i}\bigr)
	\end{multline*} 
	with $\varphi_1^\mathsf T \equiv 1$. 
\end{lemma} 

\begin{proof}
	\blue{Notice that the analogue of Lemma~\ref{lem:prec} holds in the discrete setup of subset polymers as well, in particular for the choice $D_0:=\{x\}$ and $D_1:=D'$. However, since two disjoint polymers $Y_1$ and $Y_2$ cannot both intersect the same monomer $x$, only the summands corresponding to $L=\varnothing$ and $\vert L\vert=1$ in the sum on the right-hand side of the identity in Lemma~\ref{lem:prec} provide non-trivial contributions.}
\end{proof} 

\blue{Let $D$ be a finite non-empty subset of $\Z^d$ and $z$ a non-negative activity function}. Set 
$$
	\tilde T(D;z):= 1+ \sum_{k=1}^\infty \frac{1}{k!}\sum_{(Y_1,\ldots, Y_k) \in \mathbb X^k}\bigl| \varphi_{1+k}^\mathsf T( D,Y_1,\ldots, Y_k) \bigr|z(Y_1) \cdots z(Y_k)
$$
and for $N\in \N$, 
\be \label{eq:tndef}
	\tilde T_N(D;z):= \mathbbm{1}_{\{\vert D\vert\leq N\}}+ \sum_{k=1}^\infty \frac{1}{k!}\sum_{(Y_1,\ldots, Y_k) \in \mathbb X^k} \1_{\{ |D| + \sum_{i=1}^k |Y_i| \leq N\}} \varphi_{1+k}^\mathsf T( D,Y_1,\ldots, Y_k) z(Y_1) \cdots z(Y_k).
\ee
Furthermore, we use the convention $\tilde T_N(\varnothing;z) =1$ for all $N\geq 1$.

\blue{Again, we lift the established recurrent relations on the level of coefficients (given by Lemma~\ref{lem:phirecurr}) to the level of partial sums and series, deriving a system of integral equations for those. The following result is an analogue of Proposition~\ref{tKSmult} and Proposition~\ref{prop:cavityint} for subset polymers.}

\begin{prop} \label{prop:delcont} Under the assumptions of Lemma~\ref{lem:phirecurr}, the identities
	$$
		\tilde T(D'\cup \{x\};z) = \tilde T(D';z) + \sum_{\substack{Y\ni x:\\ Y\cap D'=\varnothing}} z(Y) \tilde T(D'\cup Y;z)
	$$
	and for $N\in\N$,
	$$
		\tilde T_{N+1}(D'\cup \{x\};z) = \tilde T_N(D';z) + \sum_{\substack{Y\ni x:\\ Y\cap D'=\varnothing}} z(Y) \tilde T_N(D'\cup Y;z),
	$$
hold for any non-negative activity $z$.
\end{prop}

\begin{rem}
\blue{Notice that the first identity in Proposition~\ref{prop:delcont} is just a sign-flipped version of the standard Kirkwood-Salsburg equations for the reduced correlation functions found in~\cite{bissacot-fernandez-procacci2010}.}
\end{rem}

\begin{proof}
	Lemma~\ref{lem:phirecurr} yields 
	\begin{align*}
		&\1_{\{ |D'\cup\{x\}| + \sum_{i=1}^k |Y_i| \leq N+1\}} \varphi_{1+k}^\mathsf T( \red{D'\cup\{x\}},Y_1,\ldots, Y_k) \\
	&\qquad  =	\1_{\{ |D'| + \sum_{i=1}^k |Y_i| \leq N\}} \varphi_{1+k}^\mathsf T(D', Y_1,\ldots, Y_k) \\
	&\qquad \qquad 	+  \sum_{i=1}^k \bigl( - \1_{\{Y_i \ni x,\ Y_i \cap D'=\varnothing \}}\bigr) 	\1_{\{ |D'\cup Y_i| + \sum_{j\neq i} |Y_j| \leq N\}} \varphi_{k}^\mathsf T( D'\cup Y_i, (Y_j)_{j\neq i}).
	\end{align*}
	The proof of the recurrence relation for $\tilde T_N(\cdot;z)$ is concluded by exploiting the alternating sign property  of the Ursell functions,  summing over $k$ and  $Y_1,\ldots, Y_k$,  and exploiting the symmetry of $\varphi_k^\mathsf T$. The recurrence relation for $\tilde T(\cdot;z)$ follows by passing to the limit $N\to \infty$. 
\end{proof} 

\begin{proof}[Proof of Theorem~\ref{thm:bfp}]
To prove the implication $(i)\Rightarrow (ii)$, suppose that condition (i) is satisfied for some set function $a(\cdot)$. Proceeding as in the proof of Theorem \ref{thm:main1} again, we prove by induction over $N$ that
	\be \label{eq:tnd-bound}
			\tilde T_N(D;z) \leq \exp( a (D)),
	\ee
	 for all finite subsets $D\subset \Z^d$. For $N=1$, the inequality reads $\mathbbm{1}_{\{\vert D\vert\leq 1\}} \leq \exp(a(D))$ and it is true because $a(D)\geq 0$. Now, suppose it holds true for some $N\geq 1$ and all $D$. Let $\widehat{D}\subset \Z^d$ be finite. If $\widehat{D}$ is empty, then $\tilde T_{N+1}(\widehat{D};z) = 1 \leq \exp(a(\widehat{D}))$. If $\widehat{D}$ is not empty, let $x$ be any element of $\widehat{D}$ and let $D':=\widehat{D}\backslash\{x\}$. Then Proposition~\ref{prop:delcont} yields 
	$$
		\tilde  T_{N+1}(\widehat{D};z) = \tilde T_N(D';z) + \sum_{\substack{Y\ni x,\\ Y\cap D'=\varnothing}} z(Y) \tilde T_N(D'\cup Y;z). 
	$$	
	By the induction hypothesis and condition~\eqref{eq:discrete-nc}, 
	$$
		\tilde  T_{N+1}(\widehat{D};z)  \leq \e^{a(D')} +  \sum_{\substack{Y\ni x,\\ Y\cap D'=\varnothing}} z(Y) \e^{a(D'\cup Y)} \leq \e^{a(D'\cup \{x\})} = \e^{a(\widehat{D})}. 
	$$
	This completes the inductive proof of~\eqref{eq:tnd-bound}. Passing to the limit $N\to \infty$, we get $\tilde  T(D;z) \leq \exp(a(D)) <\infty$.

\noindent	To prove the converse implication $(ii)\Rightarrow (i)$, suppose that $T(D;z)$ is absolutely convergent for all finite subsets $D$. Then $\tilde T(D;z)<\infty$ and Proposition~\ref{prop:delcont}  yields 
	$$
		\tilde T(D\cup \{x\}; z) = \tilde T(D;z) + \sum_{\substack{Y\ni x,\\ Y\cap D=\varnothing}} z(Y) \tilde T(D\cup Y;z).
	$$
	Set $a(D):= \log \tilde T(D;z)$. Then $a(D)\geq 0$ because $\tilde T(D;z) \geq 1$, moreover 
	$$
		\e^{a(D\cup\{x\})} = \e^{a(D)} + \sum_{\substack{Y\ni x,\\ Y\cap D=\varnothing}} z(Y) \e^{a(D\cup Y)}
	$$
	and the inequality~\eqref{eq:discrete-nc} follows. 
\end{proof}

\red{Notice that the preceeding results of Section~\ref{sec:recc_sub} can be generalized by proving a more general version of Lemma~\ref{lem:phirecurr} --- a direct analogue of Lemma~\ref{lem:prec}, where we consider  two arbitrary finite subsets $D_0\subset\Z^d$ and $D_1\subset\Z^d$, instead of the special case where one of the subsets is a monomer. Naturally, one can view configurations of polymers not only as configurations of monomers but as configurations of disjoint snippets of arbitrary shape and derive from the generalized version of Lemma~\ref{lem:phirecurr} a system of Kirkwood-Salsburg equations different from the one in Proposition~\ref{prop:delcont}, equations which involve terms of higher order in the activity $z$. Those equations in turn lead to the following alternative for Theorem~\ref{thm:bfp}:}
\red{
\begin{thm}
	\label{thm:gen_bfp}
	Let $(z(X))_{X\in \mathbb X}$ be a non-negative activity. 
	The following two conditions are equivalent:
	\begin{enumerate}
		\item [(i)] There exists a function $a(\cdot)$ from the finite subsets of $\Z^d$ to $[0,\infty)$ such that $a(\varnothing )=0$ and the following system of inequalities is satisfied: For all finite, non-empty subsets $D\subset\mathbb{Z}^d$ there exists a subset $D_0\subset D$ such that --- setting $D_1:=D\backslash\{D_0\}$ --- we have
		\be\nonumber \label{eq:discrete-nc_gen}
			\sum\limits_{k\geq 1}\sum_{\{Y_1,\ldots,Y_k\}\subset X} z(Y_1)\ldots z(Y_k) \e^{a(D_1\cup Y_1\cup\ldots\cup Y_k) - a(D_1)} \leq \e^{a(D_1\cup D_0) - a(D_1)}-1,
		\ee
where the sum runs over sets of mutually disjoint polymers $\{Y_1,\ldots, Y_k\}\subset \mathbb X$ such that $Y_i\cap D_0\neq \varnothing$ and $Y_i\cap D_1= \varnothing$ for all $i\in\{1\ldots,k\}$.		
		\item [(ii)] $T(D;z)$ is absolutely convergent for all finite subsets $D\subset \Z^d$. 
		\end{enumerate} 
	Moreover, if one of the equivalent conditions (hence, both) holds true, then, for all finite subsets $D\subset \Z^d$, we have
	\begin{align*}\label{eq:abound_disc_gen}
		\bigl| \log T(D;z)\bigr|\leq \sum_{k=1}^\infty \frac{1}{k!}\sum_{(Y_1,\ldots, Y_k) \in \mathbb X^k}  \1_{\{\exists i:\, Y_i \cap D\neq \varnothing\}} \bigl| \varphi_{k}^\mathsf T(Y_1,\ldots, Y_k) \bigr| z(Y_1)\cdots z(Y_k) \leq a(D). 
	\end{align*}
\end{thm}
The details of the proof are left for the reader as an exercise. Notice that the sufficient condition for convergence given by Theorem~\ref{thm:gen_bfp} is more general than the one given by Theorem~\ref{thm:bfp}. However, all the proofs of sufficient conditions for systems of subset polymers in Section~\ref{sec:application} are using the special case of Theorem~\ref{thm:bfp}.
}

\section{Application to concrete hard-core models}\label{sec:application}

Our main results (Theorems~\ref{thm:main1},~\ref{thm:main2} and~\ref{thm:bfp}) provide characterizations of the domain of absolute convergence for the activity expansions $\rho_n(x_1,\ldots,x_n;z)$ from which well-known classical criteria are easily recovered (Corollaries~\ref{KP}, ~\ref{FPcrit} and~\ref{cor:GK}).  In this section, we illustrate how our convergence conditions  provide new, ``practitioner-type" sufficient conditions in concrete hard-core models, both discrete and continuous. \blue{Our goal here is not to improve on the best available conditions, but to provide upper bounds on the convergence radii that are of reasonable computational feasibility. In the one-dimensional setup of the Tonks gas, however, we are able to go as far as to recover the characterization of absolute convergence from~\cite{jansen2015tonks}.}

\subsection{Single-type subset polymers in $\Z^d$}

\blue{Consider the setup of subset polymers from Chapter~\ref{sec:sub_poly}}. Suppose there is some finite non-empty set $S\subset \Z^d$ and a scalar $z>0$ such that 
\begin{equation} \label{eq:z-single-type}
	z(X) = \begin{cases}
				z, &\quad X\text{ is a translate of }S,\\
				0, &\quad \text{otherwise}. 
			\end{cases} 
\end{equation} 
We call polymers with non-zero activity  \emph{active polymers}. Define 
$$
	V(D):= \bigl \lvert \{ X\in \mathbb X\mid z(X)>0,\, X\cap D \neq \varnothing\}\bigl\rvert,
$$
the number of active polymers intersecting a finite domain $D\subset \Z^d$. 
Notice $V(\{x\}) = |S|$, for all $x\in \Z^d$. 

\begin{thm}\label{Lgoof}
	Let $z(\cdot)$ be the activity function from~\eqref{eq:z-single-type}. 
	Suppose there exists $\alpha>0$ such that 
	\be \label{goof}
			|S|\,  \e^{\alpha V(S)} z\leq \e^{\alpha |S|}-1.
	\ee
	Then $T(D;z)$ is absolutely convergent for all finite subsets $D\subset \Z^d$, thus the activity expansions $\rho_n(x_1,\ldots,x_n;z)$ converge absolutely for all $n\in \N$ and all  $(x_1,\ldots, x_n)\in \mathbb X^n$. 
\end{thm}

\begin{rem}\blue{Notice that Theorem~\ref{Lgoof} improves on the upper bounds for the convergence radii given by Koteck{\'y}-Preiss and by Gruber-Kunz. The improvement over the Gruber-Kunz condition is achieved by a more sophisticated choice of the ansatz function $a$ in the proof of the theorem. However, although we do not have a general proof that the result by Fern{\'a}ndez-Procacci is stronger, notice that for \red{all non-pathological examples we considered (e.g., non-overlapping dimers or cubes)} Fern{\'a}ndez-Procacci provides better bounds than Theorem~\ref{Lgoof}.}
\end{rem}

\begin{ex} [Hypercubes]
	If $S= \{1,\ldots, k\}^d$ with $k\in \N$, condition~\eqref{goof} becomes 
	$$
		z\leq \sup_{\alpha > 0} \frac{\exp(\alpha k^d )-1}{k^d \exp( \alpha (2k-1)^d)}.
	$$
	Carrying out the optimization over $\alpha$ yields the condition 
	$$
		(2k-1)^d z \leq \Bigl(1-\frac{1}{(2 - 1/k)^d}\Bigr)^{(2-1/k)^d -1}. 
	$$
	In the limit $d\to \infty$ at fixed $k\geq 2$, the right-hand side converges from above to the familiar bound $1/\e$.
\end{ex}

\begin{proof} [Proof of Theorem~\ref{Lgoof}]
	We apply Theorem~\ref{thm:bfp} with $a(D):= \alpha V(D)$. We check that  $V(\cdot)$ is strongly subadditive. Let $B,C$ be finite subsets of $\Z^d$. Then for every polymer $X$,
	$$
		\1_{\{X\cap B\neq \varnothing \}} + 	\1_{\{X\cap C\neq \varnothing \}}
			\geq 		\1_{\{X\cap (B\cup C)\neq \varnothing \}} + 	\1_{\{X\cap (B\cap C) \neq \varnothing \}}.
	$$
	Indeed if $X$ intersects $B$ but not $C$ (or $C$ but not $B$), the inequality reads $1+ 0 \geq 1 +0$ and it is true. If $X$ intersects both $B$ and $C$, the inequality reads $1+1 \geq 1+ \1_{\{X\cap (B\cap C) \neq \varnothing\}}$ and it is true as well. Finally if $X$ intersects neither $B$ nor $C$, then both sides of the inequality vanish. Summing over all polymers $X$, we get 
	\be \label{eq:V2subadditive}
		V(B) + V(C) \geq V(B\cup C) + V(B\cap C). 
	\ee
	Now we turn to the criterion (i) from Theorem~\ref{thm:bfp}. 
	Condition~\eqref{eq:discrete-nc} for $D'= \varnothing$ reads 
	$$
		|S|\, z\, \e^{\alpha V(S)} \leq \e^{\alpha V(\{x\})} - 1,
	$$
	it is satisfied because of  $V(\{x\}) = |S|$ and the assumption~\eqref{goof}. For non-empty $D'$, we bound the left-hand side of condition~\eqref{eq:discrete-nc} with the help of the strong subadditivity. The inequality~\eqref{eq:V2subadditive} applied to $B = D'\cup \{x\}$ and $C= X$ yields 
	\be \label{tc1}
				V(D'\cup \{x\}) + V(X) \geq V(D'\cup X) + V(\{x\}),
	\ee
	for  $x\in \Z^d\setminus D'$, $x\in X$, and $X\cap D' =\varnothing$, and 
	\begin{align*}
		V(D'\cup X) - V(D') 
			& \leq 	V(D'\cup \{x\}) + V(X) - V(D') - V(\{x\})\\
			& =  	V(D'\cup \{x\}) - V(D') + V(S) - |S|. 
	\end{align*}
	This provides an $X$-independent bound for the exponent in the left-hand side of condition~\eqref{eq:discrete-nc}. The number of summands on the left-hand side of condition~\eqref{eq:discrete-nc} is given by the number of active polymers intersecting $x$ but not $D'$, which is equal to $V(D'\cup \{x\}) - V(D')$.  
	Thus to prove~\eqref{eq:discrete-nc} it suffices to show that 
	\be \label{eq:tc}
		\bigl(V(D'\cup \{x\}) - V(D')\bigr) z\, \e^{\alpha[V(D'\cup \{x\})  -  V(D') + V(S) - |S|]}	\leq \e^{\alpha[V(D'\cup \{x\})-V(D')]} - 1.
	\ee
	In view of condition~\eqref{goof}, the last inequality in turn follows once we check 
	$$
		\bigl(V(D'\cup \{x\}) - V(D')\bigr) \bigl( \e^{\alpha |S|} - 1\bigr) \e^{\alpha[V(D'\cup \{x\})  -  V(D') - |S|]} \leq |S|\bigl( \e^{\alpha[V(D'\cup \{x\})-V(D')]} - 1\bigr)
	$$
	or equivalently, 
	$$
		\frac{1- \exp(-\alpha |S|)}{|S|} \leq \frac{1 - \exp( -\alpha R )}{R},\quad R:= V(D'\cup \{x\}) - V(D').
	$$
	Because of the subadditivity of $V$, we have $R \geq V(\{x\}) = |S|$. The exponential map $x\mapsto \exp(- \alpha x)$ is convex and therefore the difference quotient is monotone increasing, i.e.,  $(\exp(-\alpha x) - 1)/x \leq (\exp( -\alpha y) - 1))/y$ whenever $0\leq x\leq y$. We apply the inequality to $x = R$ and $y= |S|$ and obtain the required bound. 
\end{proof}

\subsection{Single-type hard-core system in $\R^d$}
\blue{Consider a bounded convex shape $S\subset \R^d$ which is non-empty, regular closed and balanced (recall: A set $S\subset \R^d$ is called regular closed if and only if it equals the closure of its interior, i.e., $\overline{S^\circ}=S$, and it is called balanced if and only if $\alpha S\subset S$ for all $\vert\alpha\vert\leq 1$)}. We investigate the special case of the hard-core setup in the continuum from Section~\ref{sec:hardcore} where $\mathbb X$ consists of all translates $x+ S = \{x+ y\mid y\in S\}$. Let us further assume that both the activity and the reference measure $\lambda$ are translationally invariant.  Then we may identify $\mathbb X$ with $\R^d$, the reference measure $\lambda$ with the Lebesgue measure, and the activity function with a positive scalar,  $z(x) \equiv z>0$.

\blue{For an integrable function $h: \mathbb X\to\mathbb{R}$ we write 
\begin{align*}
\int_{\mathbb{X}} h(Z)\lambda(\dd Z)=\int_{\mathbb{R}^d} h(x+S)\dd x.
\end{align*}
}

 Write $|S|$ for the Lebesgue volume of the shape $S$ and define for Borel sets $D\subset \mathbb{R}^d$
\begin{align}\label{def:map_V}
V(D):= \int_{\R^d} \1_{\{(x+S)\cap D\neq \varnothing\}} \dd x. 
\end{align}
Notice $V(\{y\}) = |S|$, \blue{which is positive and finite by our assumptions on $S$}, and $V(S) = |S\oplus S|$ with $A\oplus B :=\{a+b\mid a\in A,\, b\in B\}$ the Minkowski sum. \blue{The latter identity holds since we assumed the set $S$ to be balanced which implies $\{(x+S)\cap S\neq \varnothing\}=S\oplus S$}. \red{Moreover, notice that $V$ --- as a function on $\mathbb{D}_{\varepsilon}$ with the convention $V(\varnothing)=0$ ---  satisfies the measurability assumption from Theorem \ref{thm:main2} by the same argument as formulated on p.~\pageref{arg:m} for $\tilde T$.}

We refer to such systems as \emph{single-type hard-core systems in the continuum.} In the language of stochastic geometry (\blue{see~\cite[Sections 3,4]{schneider2008stochastic}}), the associated Gibbs measure is a \emph{hard-core germ-grain model} with deterministic grain $S$ (the germs are the positions $x$).

\begin{thm}\label{cLgoof}
	Assume there exists $\alpha>0$ such that 
	\be\label{cgoof}
		  |S| \e^{\alpha V(S)} z <\e^{\alpha |S|}-1.
	\ee
	Then the activity expansions $\rho_n(x_1,\ldots,x_n;z)$ converge absolutely for all $n\in \N$ and all  $(x_1,\ldots, x_n)\in \mathbb \mathbb \mathbb X^n$. 
\end{thm}

\begin{rem}\label{rem:single_type_cont}\blue{Again, notice that  while Theorem~\ref{cLgoof} --- just as its discrete analogue Theorem~\ref{Lgoof}) --- improves on the Koteck{\'y}-Preiss condition, it is in \red{all the cases we considered as examples} weaker than Fern{\'a}ndez-Procacci (e.g., for systems of hard spheres the bounds on the radius of convergence obtained in~\cite{fernandez-procacci-scoppola2007} and~\cite{fernandez-nguyen2020} are slightly better). However, the advantage of our criterion is that an explicit bound is provided directly, with no need for numerical computation regardless of the dimension.}
\end{rem}

\begin{ex}[Hard spheres]
	If $S=B_R(0)$ is the closed ball of radius $R>0$ around the origin,  condition~\eqref{cgoof} becomes 
	$$
		z\leq \sup_{\alpha > 0} \frac{\exp(\alpha |B_R(0)|)-1}{|B_R(0)|\, \exp( \alpha |B_{2R}(0)|)}.
	$$
	Carrying out the optimization over $\alpha$ yields the condition 
	$$
		|B_{2R}(0)|\, z \leq \Bigl(1-\frac{1}{2^d}\Bigr)^{2^d -1}. 
	$$
	In the limit $d\to \infty$ at fixed $R>0$, the right-hand side converges from above to the familiar bound $1/\e$.

\end{ex}

\begin{proof}
Let $\alpha>0$ satisfy condition~\eqref{cgoof}. Set $a(\widehat D):= \alpha V(\widehat D)$ for $\widehat D\in\mathbb D_\varepsilon $, choose some chopping map $C$ and let $D\in \mathbb D_\varepsilon$ (the snippet-size $\varepsilon>0$ will be specified later in the proof). For the simplest selection rule $s$ choosing always the first snippet $E_1$, condition (i) in Theorem~\ref{thm:main2} reads 
\be 	\label{smth}
	\sum_{k=1}^\infty \frac{z^k}{k!}\int_{\mathbb{X}^k} I(E_1;D';Y_1,\ldots, Y_k)\, \e^{\alpha [V(D'\cup (\cup_{i=1}^k Y_i))-V(D')]}\lambda^k(\dd\vect Y) \leq 
 \e^{\alpha[V(D'\cup E_1)-V(D')]}-1,
\ee
where $D'=D\backslash E_1$.

Notice that --- unlike in the discrete case ---  terms of order higher than one in $z$ do not necessarily vanish in the series in \eqref{smth}. 
Inspired by the proof of Theorem~\ref{cgoof}, we try first to bound the exponent on the left-hand side of~\eqref{smth}, seeking a bound that separates $Y:= Y_1\cup \cdots \cup Y_k$  from $E_1$ and $D'$. If the constraint were that $E_1\subset Y$, we would conclude with strong subadditivity applied to $B:= Y$ and $C:= D'\cup E_1$ that $V(E_1)+ V(D'\cup Y) \leq V(Y)+ V(D'\cup E_1)$. For the weaker constraint $E_1\cap Y\neq \varnothing$, this is no longer true. \blue{Let $Z\in\mathbb X$.} A straightforward case distinction reveals that under the indicator $I$, the inequality 
$$
	\1_{\{ Z\cap E_1\neq \varnothing\}} +  	\1_{\{ Z\cap (D'\cup Y)\} \neq \varnothing\}} \leq  	\1_{\{ Z\cap Y\neq \varnothing\}} +  	\1_{\{ Z\cap (D'\cup E_1)\neq \varnothing\}}   
$$
is correct for all possible values of the left-hand side except possibly $1+1$. 
Indeed it may happen that $Z$ intersects both $E_1$ and $D'\cup Y$, hence a fortiori $D'\cup E_1$, but not $Y$, so that the right-hand side becomes $0+1$. 
This happens precisely when $Z$ intersects $D'$ and $E_1$ but not $ Y$. The inequality becomes correct if we add the indicator of this event to the right-hand side. Integrating over $Z$, we obtain 
$$
	V(E_1)+ V(D'\cup Y) \leq V(Y)+ V(D'\cup E_1) + \int_\mathbb X \1_{\{Z\cap E_1\neq \varnothing,\, Z\cap Y=\varnothing,\, Z\cap D'\neq \varnothing\}} \lambda(\dd Z).
$$
Moreover, there exists a constant $C=C(S,d)>0$ that depends only on the dimension $d$ and the shape $S$ such that if $k=1$ and $Y=Y_1\in \mathbb X$ is a translate of $S$, then 
$$
	\int_\mathbb X \1_{\{Z\cap E_1\neq \varnothing,\, Z\cap Y=\varnothing,\, Z\cap D'\neq \varnothing\}} \lambda(\dd Z)\leq C\eps.
$$
Indeed, on the left side we may drop the indicator that $Z$ intersects $D'$ and see that it is sufficient to check 
\be \label{eq:grainbound1}
	\int_\mathbb X \1_{\{Z\cap E_1\neq \varnothing,\, Z\cap Y=\varnothing\}} \lambda(\dd Z)\leq C\eps.
\ee
\blue{To see that such an estimate holds let $B_\varepsilon(c)$ be a closed ball of radius $\varepsilon$ around some $c\in\R^d$ containing the snippet $E_1$ and let $x\in Y\cap E_1$ (by assumption this intersection is non-empty). Then $x\in B_\varepsilon(c)$ and the inequality 
\[
\1_{\{Z\cap E_1\neq \varnothing,\, Z\cap Y=\varnothing\}}\leq\1_{\{Z\cap B_\varepsilon(c)\neq \varnothing,\, x\notin Z\}}
\]
holds pointwise in $Z$, thus also 
\[
\int_\mathbb X \1_{\{Z\cap E_1\neq \varnothing,\, Z\cap Y=\varnothing\}}\lambda(\dd Z)\leq
\int_\mathbb X \1_{\{Z\cap B_\varepsilon(c)\neq \varnothing,\, x\notin Z\}}\lambda(\dd Z).
\]
Notice that  $$\int_\mathbb X \1_{\{Z\cap B_\varepsilon(c)\neq \varnothing,\, x\notin Z\}}\lambda(\dd Z)=\vert\{y\in\R^d\vert (y+S)\cap B_\varepsilon(0)\neq \varnothing, \tilde{x}\notin y+S\}\vert,$$ where $B_\varepsilon(0)$ is the closed ball of radius $\varepsilon$ around $0$ and $\tilde{x}:=x-c$. For the set on the right-hand side of that equation, the identity
\[
\{y\in\R^d\vert (y+S)\cap B_\varepsilon(0)\neq \varnothing, \tilde{x}\notin y+S\}=\left(S\oplus B_\varepsilon(0)\right)\backslash (\tilde{x}+S),
\]
holds since $S$ being balanced directly implies \[\{y\in\R^d\vert (y+S)\cap B_\varepsilon(0)\neq \varnothing\}=S\oplus B_\varepsilon(0)\] and \[\{y\in\R^d\vert\tilde{x}\in y+S\}=\tilde{x}+S.\]
Furthermore, observe that the inclusion
\[\left(S\oplus B_\varepsilon(0)\right)\backslash (\tilde{x}+S)\subset \left(\left(S\oplus B_\varepsilon(0)\right)\backslash S\right)\cup \left(\left(S\oplus B_\varepsilon(\tilde{x})\right)\backslash (\tilde{x}+S)\right)\]
holds since $\tilde{x}\in B_\varepsilon(0)$ and $S$ is balanced set.
Moreover, $\left(S\oplus B_\varepsilon(\tilde{x})\right)\backslash (\tilde{x}+S)$ is the translate of $\left(S\oplus B_\varepsilon(0)\right)\backslash S$ by $\tilde{x}$, hence it has the same Lebesgue volume and
\begin{align*}
\vert\left(S\oplus B_\varepsilon(0)\right)\backslash (\tilde{x}+S)\vert\leq&\vert\left(S\oplus B_\varepsilon(0)\right)\backslash S\vert+ \vert\left(S\oplus B_\varepsilon(\tilde{x})\right)\backslash (\tilde{x}+S)\vert\\=&2\vert\left(S\oplus B_\varepsilon(0)\right)\backslash S\vert\\=&2\left(\vert S\oplus B_\varepsilon(0)\vert-\vert S\vert\right).
\end{align*}
Finally, by Steiner's formula for compact convex  sets (see~\cite{schneider2008stochastic}), $\vert S\oplus B_\varepsilon(0)\vert-\vert S\vert$ is given by a non-constant polynomial in $\varepsilon$, which yields a bound of the form given by the right-hand side of \eqref{eq:grainbound1}  (where the constant $C>0$ can be expressed in terms of the intrinsic volumes of $S$ following the formula).
}

Consequently, we obtain the bound
\be \label{tc1_cont}
	V(E_1)+ V(D'\cup Y) \leq V(Y)+ V(D'\cup E_1) + C\eps 
\ee
 which corresponds to the bound \eqref{tc1} in the proof of Theorem~\ref{Lgoof}. 

The inequality~\eqref{tc1_cont} immediately yields the following upper bound for the left-hand side of~\eqref{smth}:
$$
	\e^{\alpha C \eps}\, \e^{\alpha [V(D'\cup E_1)-V(D')- V(E_1)]} \sum_{k=1}^\infty \frac{z^k}{k!}\int_{\mathbb{X}^k} I(E_1;D';Y_1,\ldots, Y_k)\, \e^{\alpha V(Y)}\lambda^k(d\vect Y).
$$
The summand for $k=1$ is equal to 
$$
	z\, \e^{\alpha V(S)}\int_\mathbb X\mathbbm{1}_{\{Y_1\cap E_1\neq \varnothing, Y_1\cap D'= \varnothing\}}\lambda(\dd Y_1)=z[V(D'\cup E_1)-V(D')]\e^{\alpha V(S)}.
$$
For $k\geq 2$, we bound $V(Y) \leq \sum_{i=1}^k V(Y_i) = k V(S)$, drop the indicator that the $Y_i$'s do not intersect $D'$, and get the upper bound 
$$
	z^k \e^{\alpha k V(S)} \int_{\mathbb X^k}\prod_{i=1}^k \1_{\{Y_i\cap E_1 \neq \varnothing\}}\1_{\{Y_1,\ldots, Y_k\text{ disjoint} \}}\lambda^k (\dd \vect Y).
$$
Notice that there exists $N\in \N$ such that for all $k\geq N+1$ the integral vanishes. \blue{To see this, assume that there are infinitely many disjoint objects $Y\in\mathbb X$ intersecting the snippet $E_1$ (and therefore some open $\varepsilon$-ball $B_\varepsilon$ in which the snippet is contained). Since all the objects $Y$ are translates of $S$,  we can choose the same radius $r>0$ for all of the infinitely many disjoint objects $Y$ intersecting $E_1$ such that $Y=x+S \subset B_r(x)$. Naturally, every such $r$-ball must intersect $B_\varepsilon$ and therefore their union is again a bounded Borel subset of $\mathbb{R}^d$. But --- by our assumptions on the shape $S$ --- every $Y\in\mathbb X$ has the same fixed, strictly positive Lebesgue measure, thus their disjoint union must have infinite Lebesgue measure, which is a contradiction to its boundedness.} 

For $k \leq N$, we drop the indicator that $Y_3,\ldots, Y_k$ are disjoint and find that the integral is bounded by 
$$
	V(E_1)^{k-2} \int_{\mathbb X}\1_{\{Y_1\cap E_1 \neq \varnothing\}} \Bigl( \int_\mathbb X \1_{\{Y_2\cap E_1 \neq \varnothing,\, Y_2 \cap Y_1 = \varnothing\}}\lambda_z(\dd Y_2) \Bigr) \lambda_z(\dd Y_1). 
$$
The inner integral is bounded by $C\eps$ because of~\eqref{eq:grainbound1}, the outer integral gives an additional factor $V(E_1)$. 
Altogether, the left-hand side of~\eqref{smth} is bounded by 
$$
	\e^{\alpha C \eps}\, \e^{\alpha [V(D'\cup E_1)-V(D')- V(E_1)]}\Bigl( z[V(D'\cup E_1)-V(D')]\e^{\alpha V(S)}+  C\eps \sum_{k=2}^N z^k V(E_1)^{k-1} \e^{\alpha k \red{V(S)}} \Bigr). 
$$
Proceeding as in the proof of Theorem~\ref{goof}, but taking into account the strict inequality from assumption~\eqref{cgoof}, we find that there exist $\alpha>0$ such that 
$$
	 \e^{\alpha [V(D'\cup E_1)-V(D')- V(E_1)]} z[V(D'\cup E_1)-V(D')]\e^{\alpha V(S)} <\e^{\alpha [V(D'\cup E_1) - V(D)] }- 1,
$$
compare to~\eqref{eq:tc}. Therefore, picking $\eps$ small enough, we see that~\eqref{smth}, hence also condition (i) in Theorem~\ref{thm:main2} is satisfied and all $T(D;z)$, $D\in \mathbb D_\eps$, are absolutely convergent. The claim of the theorem follows immediately.
\end{proof}

\subsection{Multi-type hard spheres in $\R^d$}

Let $(r_n)_{n\in\N}$ be an increasing sequence of positive real numbers and let $B_{r_n}(0)\subset \R^d$, $n\in\N$, be the family of $d$-dimensional closed balls around $0$ with the corresponding radii. To the sequence $(r_n)_{n\in\N}$ of radii is associated the sequence of non-negative activities $(z_n)_{n\in\N}$ such that the ball $B_{r_i}$ has the activity $z_i$. In the setup of hard-core systems in the continuum from Section~\ref{sec:hardcore}, let $\mathbb X$ be given by all possible translates of these objects. Notice that the closed balls are compact convex sets that are non-empty and regular closed. We refer to this special case of a hard-core system as a system of \emph{multi-type hard spheres} in $\R^d$. We will show a new sufficient condition for absolute convergence of the activity expansions in these types of models.

For an integrable function $h: \mathbb X\to\mathbb{R}$ we write 
\begin{align*}
\int_{\mathbb{X}} h(Z)\lambda(\dd Z)=\sum_{\ell\geq 1}\ \int_{\mathbb{R}^d} h(x+B_{r_\ell}(0))\dd x.
\end{align*}

\blue{We define the family of functions $(V_r)_{r>0}$ by setting for  Borel sets $D\subset  \R^d$
\begin{align}\label{def:map Vr}
V_r(D):=\int_{\mathbb R^d}\1_{\{(x+B_r(0))\cap D\neq \varnothing\}} \dd x,
\end{align}
where $B_r(0)$ is the $d$-dimensional closed ball of radius $r>0$ around $0$. Naturally, the map $V_r$ coincides with the map $V$ from \eqref{def:map_V} for the grain $S$ given by the closed ball $B_r(0)$ and therefore satisfies the measurability assumption from Theorem \ref{thm:main2} (as a function on $\mathbb{D}_{\varepsilon}$ with the convention $V(\varnothing)=0$). Furthermore, we have $V_r(\{y\})=\vert B_r(0)\vert$ and $V_r(B_s(y))=\vert B_s(y)\oplus B_r(0)\vert=\vert B_{s+r}(0)\vert$ (where $\oplus$ denotes the Minkowski sum) for any $y\in\R^d$ and any numbers $r,s>0$.}

The following auxiliary result turns out to be essential for the proof of the new sufficient condition:

\begin{lemma}\label{aux_geom}
Let $D_1$ be a finite union of  bounded convex regular closed subsets of $\mathbb{R}^d$ and let $D_2$ be a $d$-dimensional ball in $\mathbb{R}^d$.
The map $(0,\infty)\ni r\mapsto \frac{V_r(D_1 \cup D_2)-V_r(D_1)}{V_r(D_2)}$ is monotonically decreasing in $r$.
\end{lemma}
\begin{proof}
First of all, observe that for sets $D$ given by a finite union of convex, regular closed subsets of $\mathbb{R}^d$ the volume $V_r(D)$ can be written as
\begin{align}\label{Steiner_conseq}
V_r(D)=\vert D\vert + S(D)r+o(r),
\end{align}
where $S(D)$ denotes the surface area of $D$. This follows from a generalized version of the classical Steiner's formula (see\red{~\cite[Section 4.4]{schneider2013})}. In particular, we see that the map $r\mapsto \frac{V_r(D_1 \cup D_2)-V_r(D_1)}{V_r(D_2)}$ is differentiable in $r=0$.
  
Next, we notice that the map satisfies the following \textit{semi-group property}: 
\begin{align*}
V_{r+\varepsilon}(D)=V_\varepsilon( D\oplus B_r(0)). 
\end{align*}

Therefore, to prove the claim of the lemma, it suffices to consider the differential at zero:
\begin{align*}
\lim\limits_{\varepsilon\searrow 0}\frac{1}{\varepsilon}\left(\frac{V_\varepsilon(A\cup B)-V_\varepsilon(A)}{V_\varepsilon(B)}-\frac{\vert A\cup B\vert-\vert A\vert }{\vert B \vert }\right),
\end{align*}
where $A:=D_1\oplus B_r(0)$ and $B:=D_2\oplus B_r(0)$.

Using the formula \eqref{Steiner_conseq}, a simple computation shows that this limit is equal to
\begin{align*}
\frac{\vert B \vert\left(S(A\cup B)-S(A)\right)-S(B)\left(\vert A\cup B\vert -\vert A\vert\right)}{\vert B \vert^2}.
\end{align*}
The monotonicity in the claim of the lemma is then equivalent to
\begin{align*}
\vert B \vert\left(S(A\cup B)-S(A)\right)-S(B)\left(\vert A\cup B\vert -\vert A\vert\right)\leq 0
\end{align*}
or, equivalently,
\begin{align*}\frac{\vert B\vert}{S(B)}
\leq \frac{\vert A\cup B\vert -\vert A\vert}{S(A\cup B)-S(A)}.
\end{align*}
Using the obvious identities $\vert A\cup B\vert-\vert A\vert = \vert B\vert - \vert A\cap B\vert$ and $S(A\cup B)-S(A)=S(B)-S(A\cap B)$, we can rewrite the last inequality as
\begin{align*}\frac{S(B)}{\vert B \vert}\leq
\frac{ S(A\cap B)}{\vert A\cap B\vert},
\end{align*}
which holds by the isoperimetric inequality since $B=D_2\oplus B_r(0)$ is a ball in $\mathbb{R}^d$ (``the ball is the shape that minimizes the surface area for given volume", see~\cite[3.2.43]{federer1969geometric}).
\end{proof}

\blue{The following sufficient condition is, in some sense, a ``continuous version" of the Gruber-Kunz criterion in the setup of hard spheres in $\mathbb{R}^d$. The similarity in the form arises as follows: To establish the recurrence relations underlying the proof of Gruber-Kunz we selected a monomer, a single point in $\Z^d$, from a configuration of polymers. We follow this idea in the proof of the following result, choosing the chopping map $C$ and the selection rule $s$ such that a tiny snippet that approximates a single point in the continuous space sufficiently well is selected. At the same time we choose an ansatz function $a$ that can be interpreted as the continuous analogue of the ansatz function from the proof of Corollary~\ref{cor:GK}.} 
\begin{thm}\label{GK_cont}
In the setup of multi-type hard objects on $\mathbb{R}^d$, assume that the activity $z$ satisfies
\begin{align}\label{gens_cont}
\exists\alpha>0:~~\sum\limits_{\ell\geq 1}\vert B_{r_\ell}\vert \e^{\alpha \vert B_{r_\ell+r_1}\vert }z_\ell< e^{\alpha \vert B_{r_1}\vert}-1,
\end{align} 
where by $\vert B_r\vert$ we denote the (Lebesgue) volume of a ball of radius $r>0$.  
Then the activity expansions $\rho_n(X_1,\ldots, X_n;z)$ converge absolutely.
\end{thm}

\begin{rem}
\red{Activity expansions for systems with infinitely many types of objects are not particularly well-studied in statistical mechanics. In the case of finitely many types, we expect our result to exceed Koteck{\'y}-Preiss but to be weaker then Fern{\'a}ndez-Procacci  --- as in the special case of a single type treated above (see Remark~\ref{rem:single_type_cont}). The general case, for $r_n\to\infty$ in particular, remains to be investigated.} 
\end{rem}

\begin{proof}
Again, our strategy is to show that condition (i) from Theorem \ref{thm:main2} is satisfied for an appropriate ansatz function $a$. By assumption, $r_1$ is the radius of the smallest ball present in the system. Set $a(\widehat D):=\alpha V_{r_1}(\widehat D)$ for $\widehat{D}\in\mathbb D_\varepsilon$, where $V_{r_1}$ is given by \eqref{def:map Vr} and $\alpha$ satisfies \eqref{gens_cont}. \blue{Choose some chopping map $C$ and let $D\in\mathbb D_\varepsilon$ (the snippet-size $\varepsilon>0$ is to be specified later in the proof).}  Just as in the proof of Theorem \ref{cLgoof}, \blue{independently of the choice of the snippet $E_1$ (i.e., independently of the selection rule $s$)}, we obtain the following upper bound for the left-hand side of the inequality from condition (i):
\begin{align}\label{bound:hsph1}
\e^{\alpha C_1 (\eps)}\, \e^{\alpha [V(D'\cup E_1)-V(D')- V(E_1)]} \sum_{k=1}^\infty \frac{z^k}{k!}\int_{\mathbb{X}^k} I(E_1;D',Y_1,\ldots, Y_k)\, \e^{\alpha V(\vect Y)}\lambda^k(\dd\vect Y).
\end{align}
The positive number $C_1(\varepsilon)$ converges towards $0$ for $\varepsilon\searrow 0$ and is precisely the bound from \eqref{eq:grainbound1} in the proof of Theorem \ref{cLgoof} for $Y$ given by a translate of $B_{r_1}(0)$, i.e., by a sphere of minimal volume present in the system.

The summand for $k=1$ in~\eqref{bound:hsph1} is equal to
\begin{align*}
\sum\limits_{\ell\geq 1}z_\ell \e^{\alpha \vert B_{r_\ell+r_1}\vert}\int_{\mathbb X}\mathbbm{1}_{\{r(Y_1) = r_\ell \}}\mathbbm{1}_{\{Y_1\cap E_1\neq \varnothing, Y_1\cap D'= \varnothing\}}\lambda(\dd Y_1).
\end{align*}
Notice that the integrals in the last expression are equal to $[V_{r_\ell}(D'\cup E_1)-V_{r_\ell}(D')]$ for every $\ell\in\mathbb N$.

The summand for any $k\geq 2$ in~\eqref{bound:hsph1} is bounded from above by
\begin{align}\label{mt_hot}
\sum\limits_{\ell_1,\ldots,\ell_k}z_{\ell_1}\ldots z_{\ell_k}\e^{\alpha\sum\limits_{i=1}^k \vert B_{r_{\ell_i}+r_1}\vert}\int_{\mathbb X^k}\prod\limits_{i=1}^k\mathbbm{1}_{\{r(Y_i) = r_{\ell_i} \}}\prod\limits_{i=1}^k\mathbbm{1}_{\{Y_i\cap E_1\neq \varnothing, Y_i\cap D'= \varnothing\}}\mathbbm{1}_{\{Y_1,\ldots,Y_k \textit { disjoint}\}}\lambda^k(\dd \vect Y),
\end{align}
which --- by arguments similar to the ones used for the bound $C_1(\varepsilon)$ --- is again bounded by
\begin{align}\label{mt_hot_bound}
C_2(\varepsilon)\sum\limits_{\ell_1,\ldots,\ell_k}z_{\ell_1}\ldots z_{\ell_k}\e^{\alpha\sum\limits_{i=1}^k \vert B_{r_{\ell_i}+r_1}\vert}\vert B_{r_{\ell_1}}\vert\ldots\vert B_{r_{\ell_k}}\vert
\end{align}
for a positive constant $C_2(\varepsilon)$ that is independent of $k$ and satisfies $C_2(\varepsilon)\searrow 0$ for $\varepsilon\searrow 0$.
Notice that the sum in \eqref{mt_hot_bound} is finite for every $k\in\mathbb{N}$ by assumption \eqref{gens_cont} (since it is simply given by the $k$-th power of the left-hand side of the inequality in \eqref{gens_cont}). Moreover, by the same argument as in the proof of Theorem \ref{cLgoof}, the expression in \eqref{mt_hot} does vanish for all but finitely many $k\in \mathbb{N}$, i.e., there exists a number $N\in\mathbb{N}$ such that \eqref{mt_hot} is equal to zero for all $k\geq N+1$.

Altogether we get the upper bound
\begin{align*}
&\e^{\alpha C_1 (\eps)}\, \e^{\alpha [V(D'\cup E_1)-V(D')- V(E_1)]}\times\Bigl( \sum\limits_{\ell\geq 1}z_\ell\e^{\alpha \vert B_{r_\ell+r_1}\vert}[V_{r_\ell}(D'\cup E_1)-V_{r_\ell}(D')]+\\& C_2(\varepsilon)\sum\limits_{2\leq k\leq N}\sum\limits_{\ell_1,\ldots,\ell_k}z_{\ell_1}\ldots z_{\ell_k}\e^{\alpha\sum\limits_{i=1}^k \vert B_{r_{\ell_i}+r_1}\vert}\vert B_{r_{\ell_1}}\vert\ldots\vert B_{r_{\ell_k}}\vert \Bigr). 
\end{align*}
As in the single-type case, we see that it is sufficient to prove the strict inequality 
\begin{align*}
\e^{\alpha [V(D'\cup E_1)-V(D')- V(E_1)]}\sum\limits_{\ell\geq 1}z_\ell\e^{\alpha \vert B_{r_\ell+r_1}\vert}[V_{r_\ell}(D'\cup E_1)-V_{r_\ell}(D')]< \e^{\alpha [V(D'\cup E_1)-V(D)]}-1
\end{align*}
for small values of $\varepsilon>0$ and, consequently, for small volumes of the snippet $E_1$ contained in an $\varepsilon$-ball.

To do so, we bound $[V_{r_\ell}(D'\cup E_1)-V_{r_\ell}(D')]$ from above by $[V_{r_\ell}(D'\cup B_\varepsilon)-V_{r_\ell}(D')]$ for every $\ell\in\mathbb{N}$, where $B_\varepsilon$ is the ball of radius $\varepsilon$ containing the snippet $E_1$.
Then we use Lemma \ref{aux_geom} to obtain
\begin{align*}
[V_{r_\ell}(D'\cup B_\varepsilon)-V_{r_\ell}(D')]\leq [V_{r_m}(D'\cup B_\varepsilon)-V_{r_m}(D')]\frac{V_{r_\ell}(B_\varepsilon)}{V_{r_m}(B_\varepsilon)}
\end{align*}
for $m\leq \ell$ (since in that case $r_m\leq r_\ell$ holds by assumption) and therefore
\begin{align*}
\sum\limits_{\ell\geq 1}z_\ell \e^{\alpha \vert B_{r_\ell+r_1}\vert}[V_{r_\ell}(D'\cup E_1)-V_{r_\ell}(D')]\leq \frac{V_{r_1}(D'\cup B_\varepsilon)-V_{r_1}(D')}{V_{r_1}(B_\varepsilon)}\sum\limits_{\ell\geq 1}z_\ell\e^{\alpha \vert B_{r_\ell+r_1}\vert}V_{r_\ell}(B_\varepsilon).
\end{align*}
By dominated convergence and assumption \eqref{gens_cont} we can choose $\varepsilon>0$ small enough to strictly bound the right-hand side of the last equation by
\begin{align*}
\frac{V_{r_1}(D'\cup E_1)-V_{r_1}(D')}{V_{r_1}(E_1)}(\e^{\alpha V_{r_1}(E_1)}-1).
\end{align*}
Finally, it suffices to show the inequality
\begin{align}
\e^{\alpha [V(D'\cup E_1)-V(D')- V(E_1)]}\frac{V_{r_1}(D'\cup E_1)-V_{r_1}(D')}{V_{r_1}(E_1)}(\e^{\alpha \vert V_{r_1}(E_1) \vert}-1)\leq \e^{\alpha [V(D'\cup E_1)-V(D)]}-1
\end{align}
as in \eqref{eq:tc} to conclude the proof.
\end{proof}

\subsection{Tonks gas on $\Z$}

\blue{Next we turn to the discrete one-dimensional Tonks gas with translationally invariant activities.} That is, in the setup of subset polymers from Section~\ref{sec:sub_poly} for $d=1$, let $(z_\ell)_{\ell\in \N}$ be a sequence of non-negative numbers and consider the activity
\be 	\label{eq:z-tonks} 
	z(X) = \begin{cases} 
			z_\ell, &\quad X = \{m, m+1,\ldots, m+\ell-1\}\text{ for some }m \in \Z,\\
			0, &\quad \text{else}. 
		\end{cases} 
\ee

\begin{thm}\label{Tres}
	Let $d=1$ and let $(z_\ell)_{\ell\in \N}$ be a sequence of non-negative activities. 
	\begin{enumerate} 
		\item [(a)]  Suppose there exists $\alpha>0$ such that 
	\be \label{onedTsuff}
		\sum_{\ell=1}^\infty \e^{\alpha \ell}z_\ell\leq \e^{\alpha}-1.
	\ee
	Then $T(D;z)$ is absolutely convergent for all finite subsets $D\subset \Z$. 
		\item[(b)] Conversely, if $T(D;z)$ is absolutely convergent for all finite subsets $D\subset \Z$, then there exists $\alpha>0$ such that~\eqref{onedTsuff} holds true.
	\end{enumerate} 
\end{thm} 

\begin{rem}\blue{The condition~\eqref{onedTsuff} is exactly the necessary and sufficient criterion for absolute convergence of the activity expansion of the pressure in the system derived in~\cite{jansen2015tonks}. While the result itself is not novel, we consider the proof to be instructive since it demonstrates how our approach can provide conditions improving on the Fern{\'a}ndez-Procacci criterion. In this concrete setup even the optimal result --- recovering the whole domain of convergence --- can be achieved. }
\end{rem}
The proof of the sufficient condition relies on a refinement of Theorem~\ref{thm:bfp}. Roughly, we weaken condition (i) to consider the Kirkwood-Salsburg inequalities being satisfied only for single rods rather than for arbitrary configurations of rods; at the same time we specify the selection rule by assuming \blue{that the leftmost (or, alternatively, the rightmost) element $\{x\}$ is always picked from any given domain}.

\begin{prop} \label{prop:selection}
	Suppose there exists a non-negative function $a(\cdot)$ from the finite intervals of $\Z$ to $[0,\infty)$ with $a(\varnothing) =0$ and for every finite interval $D$ of $\mathbb Z$ with $x=\min D$ such that
	\be \label{eq:selection}
		\sum_{\substack{Y\ni x,\\ Y\cap D' =\varnothing}} z(Y)\, \e^{a(D'\cup Y) - a(D')} \leq \e^{a(D'\cup \{x\}) - a(D')} - 1,
	\ee
	where we set $D'=D\backslash\{x\}$. Then $T(D,z)$ is absolutely convergent, for all finite $D\subset \Z$  (interval or not).
\end{prop}

\begin{proof}
	\blue{We revisit the proof of the implication (i)$\ \Rightarrow\ $(ii) of Theorem~\ref{thm:bfp} given on p.~\pageref{eq:tnd-bound} and prove first by induction over $N$ that $\tilde T_N(D;z)\leq \exp(a(D))$, for all finite discrete intervals $D\subset \Z$. For $N=1$, the inequality is trivial because $\tilde T_1(D;  z)=\mathbbm{1}_{\{\vert D\vert \leq 1\}}\leq 1\leq\exp(a(\widehat D))$. Now, suppose $\tilde T_N(D;z)\leq \exp(a(D))$ for some $N\in\N$ and all discrete intervals $D\subset \Z$. Let $\widehat D\subset \Z$ be any discrete interval. If $\widehat D=\varnothing$, then $\tilde  T_{N+1}(\widehat D;z) =1\leq \exp(a(\widehat D))$. If $\widehat D$ is non-empty, let $x:=\min \widehat D$ (or, alternatively, $x:=\max \widehat D$) and set $D'=\widehat D\backslash\{x\}$, then Proposition~\ref{prop:delcont} yields
$$
\tilde T_{N+1}(\widehat D;z) = \tilde T_N(D';z) + \sum_{\substack{Y\ni x:\\ Y\cap D'=\varnothing}} z(Y) \tilde T_N(D'\cup Y;z).
$$ Since all the arguments $D'$ and $D'\cup Y$ of $\tilde T_N$ on the right side of this identity are again finite discrete intervals, the inductive hypothesis and our assumption \eqref{eq:selection} imply that
$$\tilde T_{N+1}(\widehat D;z)\leq \e^{a(D')} + \sum_{\substack{Y\ni x:\\ Y\cap D'=\varnothing}} z(Y) \e^{a(D'\cup Y)}  \leq \e^{ a(\widehat D)}.$$ 
This completes the inductive proof of the inequality $\tilde T_N(D;z) \leq \exp( a(D))$. 
	Passing to the limit $N\to \infty$, we get $\tilde T(D;z) \leq \exp(a(D))<\infty$ for all intervals $D\subset \Z$. The convergence extends to all finite sets because $\log \tilde T(\cdot;z)$ is subadditive.}
\end{proof} 

\begin{proof}[Proof of Theorem~\ref{Tres}(a)] 
	Consider the selection rule $s(D):= \min D$ that picks the left-most point of a finite set. For $\alpha>0$ and $L \in \N$ let 
	$$
		V_L(D):= \bigl|\{ X\subset \Z\mid X\text{ is an $L$-rod, } X\cap D \neq  \varnothing\}\bigr|
	$$
	and $a(D)\equiv a_{\alpha,L}(D):= \alpha\, V_L(D)$. The choice of  $\alpha$ and $L$ is specified later. For a non-empty interval $D$, write  $x:= s(D) = \min (D)$, and $D':= D\setminus \{x\}$. If $D'$ is non-empty, then condition~\eqref{eq:selection} reads
	\be \label{eq:selection2}
		\sum_{\ell=1}^\infty z_\ell\, \e^{\alpha \ell}  \leq \e^\alpha - 1.
	\ee
	Indeed in that case for each $\ell$ there is a single $\ell$-rod $X$ that contains $x$ but does not intersect $D'$ (note that $x+1\in D'$ because of the assumption that $D$ is an interval and $x=\min D$). The rod is simply the $\ell$-rod with right-most endpoint $x$. Moreover $V_L (D'\cup X) - V_L(D') = \ell$ and $V_L(D'\cup \{x\}) - V_L(D') = 1$. 
	
	On the other hand if $D'$ is empty, then the number of $\ell$-rods that contain any given site $x\in \Z^d$ is equal to $\ell$ and the number of $L$-rods intersecting an $\ell$-rod is equal to $L+\ell-1$, therefore 
	condition~\eqref{eq:selection} reads instead 
	\be \label{eq:selection3}
		\sum_{\ell=1}^ \infty  \ell z_\ell\, \e^{\alpha( \ell+L-1)} \leq \e^{\alpha L} - 1. 
	\ee
	The proof of Theorem~\ref{Tres} is complete once we check the existence of $\alpha>0$ and $L\in \N$ such that the inequalities~\eqref{eq:selection2} and~\eqref{eq:selection3} hold true. Set 
	$$
		h(u):= 1+ \sum_{\ell=1}^\infty z_\ell\, u^\ell\qquad (u\in \R_+).
	$$
	Conditions~\eqref{eq:selection2} and~\eqref{eq:selection3} are equivalent to 
	\be \label{eq:selection4} 
		h(\e^\alpha) \leq \e^\alpha, \quad h'(\e^\alpha) \leq 1- \e^{-\alpha L}. 
	\ee
	Notice that $h$ is convex and monotone increasing with $h(0)=1$. The assumption~\eqref{onedTsuff} yields the existence of some $u = \e^\alpha >0$ such that~$h(u)< u$. 	On the other hand, clearly $h(0) =1>0$. Therefore the mean-value theorem yields the existence of a point $\tilde u \in (0,u)$ such that $h(\tilde u) = \tilde u$. The point $\tilde u$ is necessarily larger then $1$ because $h(\tilde u)$ is.  Suppose by contradiction that $h'(\tilde u) \geq 1$. Then the convexity of $h$ implies 
	$$
		h(u) \geq h(\tilde u)  + h'(\tilde u) (u-\tilde u) \geq h(\tilde u)  +  (u-\tilde u) = u,
	$$
	which contradicts the assumption $h(u)<u$. Therefore $h(\tilde u) =\tilde u>1$ and $h'(\tilde u)<1$. Replacing $\alpha$ with $\tilde \alpha:= \log \tilde u$ if needed, and picking $L=L(\alpha)$ large enough, we find that ~\eqref{eq:selection4} is satisfied for some $\alpha>0$. This concludes the proof. 
\end{proof} 

\begin{proof}[Proof of Theorem~\ref{Tres}(b)]
	Let $a(D):=\log \tilde T(D;z) = \log T(D;-z)$. In view of Eq.~\eqref{eq:tdz} and the alternating sign property,  we have 
	$$
		a(D) = \sum_{k=1}^\infty \frac{1}{k!}\sum_{(Y_1,\ldots, Y_k)\in \mathbb X^k} \1_{\{\exists i:\, Y_i \cap D\neq \varnothing\}} \bigl|\varphi_k^\mathsf T(Y_1,\ldots, Y_k) \bigr|\, z(Y_1) \cdots z(Y_k). 
	$$
	 \blue{By Proposition~\ref{prop:delcont}}, for every $D\subset \Z\setminus \{1\}$, we have 
	\be \label{eq:tonksnec1}
		\sum_{\substack{Y\ni 1, \\ Y\cap D= \varnothing}} z(X) \e^{a(D\cup Y) - a(D)} \leq \e^{a(D\cup \{1\}) - a(D)} - 1. 
	\ee
	Let us choose $D\subset \Z\cap (-\infty,0]$ with $0 \in D$. Then for every given $\ell\in \N$, the unique rod of length $\ell$ that contains $1$ but does not intersect $D$ is the rod $\{1,\ldots,\ell\}$, and we obtain
	\be  \label{eq:tonksnec2}
		\sum_{\ell=1}^\infty z_\ell\,  \e^{a(D\cup \{1,\ldots, \ell\}) - a(D)} \leq \e^{a(D\cup \{1\}) - a(D)} - 1. 
	\ee
	Let $D_0:=D$ and for $m\geq 1$ set  $D_{m}:= D\cup \{1,\ldots,m\}$. 
	The exponent on the left-hand side in~\eqref{eq:tonksnec2}  may be written as 
	\be  \label{eq:tonksnec3}
		a(D\cup \{1,\ldots, \ell\}) - a(D)  = \sum_{m=1}^{\ell} \bigl( a( D_{m}) -  a(D_{m-1})  \bigr).
	\ee
	Now
	\begin{multline*}
		a(D_m) - a(D_{m-1})\\
			 = \sum_{k=1}^\infty \frac{1}{k!}\sum_{(Y_1,\ldots, Y_k)\in \mathbb X^k}\Bigl( \1_{\{\exists i:\, Y_i \cap D_m \neq \varnothing\}} - \1_{\{\exists i:\, Y_i \cap D_{m-1} \neq \varnothing\}} \Bigr)  \bigl| \varphi_k^\mathsf T(Y_1,\ldots, Y_k)\bigr| z(Y_1)\cdots z(Y_k).
	\end{multline*} 
	The only clusters $(Y_1,\ldots, Y_k)$ that contribute to the sum are those that intersect $D_{m}$ but do not intersect $D_{m-1}$. This is only possible if one of the $Y_i$'s contains $m$ and all of them are contained in $\Z\cap [m,\infty)$. Thus 
	\begin{multline*}
		a(D_m) - a(D_{m-1})\\
			= \sum_{k=1}^\infty \frac{1}{k!}\sum_{(Y_1,\ldots, Y_k)\in \mathbb X^k} \1_{\{\exists i:\, Y_i \ni m\}} \1_{\{\forall i:\, Y_i \subset [m,\infty)\}}  \bigl| \varphi_k^\mathsf T(Y_1,\ldots, Y_k)\bigr| z(Y_1)\cdots z(Y_k).
	\end{multline*} 
	Because of the translational invariance, the value of the sum does not depend on $m$. Thus $a(D_m) - a(D_{m-1}) = \alpha>0$ for all $m\geq 1$ and some $\alpha >0$. Turning back to~\eqref{eq:tonksnec3}, we obtain 
	$$
		a(D\cup \{1,\ldots, \ell\}) - a(D) = \ell \alpha 
	$$
	and then~\eqref{eq:tonksnec2} yields $\sum_{\ell=1}^\infty z_\ell \exp(\alpha \ell) \leq \exp( \alpha) - 1$. 
\end{proof} 

\subsection{Tonks gas on $\R$}

Next, we want to consider the  continuous  version of the one-dimensional Tonks gas. Let $(L_\ell)_{\ell\in \N}$ be  sequence of strictly positive numbers and $\mathbb X$ the space of compact intervals $I\subset\R$ with lengths $|I|\in \{L_\ell\mid \ell\in \N\}$. The map $\R\times \N$, $(x,\ell)\mapsto [x-L_\ell/2,x+L_\ell/2]$ is a bijection between $\R\times \N$ and $\mathbb X$. The reference measure $\lambda$ is defined by the equality 
$$
	\int_\mathbb X h(X) \lambda(\dd X) = \sum_{\ell=1}^\infty \int_{-\infty}^\infty h\bigl(\bigl[ x - \tfrac{L_\ell}{2},x+ \tfrac{L_\ell}{2}\bigr]\bigr) \dd x
$$
for all non-negative measurable functions $h:\mathbb X\to \R_+$. 
We assume that the activity is of the form 
\begin{equation*}
	z(X) = \begin{cases} 
					z_\ell, &\quad X = [x,x+L_\ell]\, \text{ for some }\ell \in \N, x\in \R,\\
					0, &\quad \text{else}
			\end{cases} 
\end{equation*} 
for some sequence $(z_\ell)_{\ell\in \N}$ of non-negative numbers. We assume that rod lengths are bounded from below, i.e., there exists $\delta>0$ such that 
\be \label{eq:minlengths}
	\inf_{\ell\in \N} L_\ell \geq \delta. 
\ee

From here on, we will consider  the following chopping map: For $X=[x,x+L_\ell]\in \mathbb X$, let $C(X) = \{E_1,\ldots, E_m\}$ consist of the intersections of $X$ with the intervals $[x+(k-1)\eps, x+k\eps)$ with $k\in \Z$, where $\eps\in (0,\delta)$. 
	The space of snippets $\mathbb E_\eps$ consists of  intervals $[a,b]$ and $[a,b)$  of length $b-a\leq \eps$.

\begin{thm} \label{thm:cont1T}
	In the setup of multi-type Tonks gas on $\mathbb{R}$, under the assumption~\eqref{eq:minlengths}: 
	\begin{enumerate} 
		\item [(a)] Suppose there exists $\alpha>0$ such that 
	\be  \label{c1dTsuff}
		\sum_{\ell=1}^\infty \e^{\alpha L_\ell}z_\ell<\alpha.
	\ee
	Then the expansion for $T(D;z)$ is absolutely convergent, for all bounded sets $D\subset \R$. 
		\item[(b)] Conversely, if $T(D;z)$ is absolutely convergent for all bounded subsets $D\subset \R$, then there exists $\alpha>0$ such that~\eqref{c1dTsuff} holds true with ``$\leq$'' instead of ``$<$''.
	\end{enumerate} 
\end{thm} 

\begin{rem}\blue{The theorem essentially recovers the necessary and sufficient convergence criterion from~\cite{jansen2015tonks} (derived there for the activity expansion of the pressure in the system). The sufficient condition in \cite{jansen2015tonks} is~\eqref{c1dTsuff} with ``$\leq$'' instead of ``$<$''. Again, while the result itself is not novel, its proof demonstrates the potential of our approach to go beyond the Fern{\'a}ndez-Procacci criterion --- also in continuous setups.}
\end{rem}

\red{First we prove an auxiliary result, the analogue of Proposition \ref{prop:selection} for the continuous setup, which is not quite as trivial. We introduce the following notion: Define the $\varepsilon$-gap-filling operation $\ \widehat{\cdot}\ $ by setting $\widehat{D}:=D\cup\{x\in\R\vert\ \exists y, z\in D \textit{ with } y<x<z \textit{  such that } z-y\leq \varepsilon \}$ for any $D\subset \R$. Let $\mathscr P$ be some subset of the power set of $\R$, we say that a function $\xi:\mathscr P\to\R$ does not see gaps of diameter at most $\varepsilon$ if it is invariant under the $\varepsilon$-gap-filling operation, i.e., if $\xi(D)=\xi(\widehat D)$ for all $D\in \mathscr P$.}

\begin{prop} \label{prop:cont_selection}
Suppose that there exists a non-negative, measurable map $a(\cdot)$ defined on finite unions of (bounded) intervals which \red{does not see gaps of diameter at most $\varepsilon$}  and satisfies the following system of inequalities: For any (bounded) interval $D$ with $C(D)=\{E_1,\ldots, E_n\}$, $E_1,\ldots, E_n\in \E_\varepsilon$, where the chopping map $C$ is defined as above, there is a subinterval $E_s\subset D$ of length at most $\varepsilon$, such that
\begin{align}\label{eq:cont_selection}
			\sum_{k=1}^\infty \frac{1}{k!}\int_{\mathbb X^k} I(E_s; D'; Y_1,\ldots, Y_k) \e^{a(D'\cup Y_1\cup \cdots \cup Y_k) - a(D')} \lambda_z^k(\dd \vect Y) 
			\leq \e^{a(E_s\cup D') - a(D')}-1,
		\end{align}
where we set $D':=D\backslash E_s$ and $I(E_s;D'; Y_1,\ldots, Y_k)$ is the indicator from Eq.~\eqref{eq:idef}. Then $T(D;z)$ is absolutely convergent for all bounded subsets $D\subset \mathbb R$.
\end{prop}

\begin{proof}
We can modify the Kirkwood-Salsburg-type \red{equations} $\tilde{\kappa}^s_z$ from Chapter~\ref{sec3.3} as follows:  If $\xi(\cdot)$ is a function from $\mathbb D_\eps$ to $\R_+$ \blue{that does not see gaps of diameter at most $\varepsilon$ and satisfies the measurability assumption from Theorem~\ref{thm:main2}}, define the function $\tilde{\mathscr{K}}^s_z\xi$ (possibly assuming the value ``$\infty$") by
\begin{multline*}
		(\tilde{\mathscr{K}}^s_z\xi)\bigl( D\bigr)
			:= \1_{\{n\geq 2\}}\, \xi\Bigl(\reallywidehat{E'_2\cup\ldots\cup E'_n}\Bigr) \\
				+\sum_{k=1}^\infty \frac{1}{k!} \int_{\mathbb X^k} I(E_s;E'_2\cup \cdots \cup E'_n; Y_1,\ldots,Y_k) \xi\Bigl(\reallywidehat{E'_2\cup\ldots\cup E'_n\cup Y_1\cup\ldots\cup Y_k} \Bigr)\lambda_z^k(\dd \vect Y),
\end{multline*}
for $D\in \mathbb D_\eps$ with $E_1,\ldots,E_n\subset \mathbb E_\eps$ and $C(D)=\{E_1,\ldots, E_n\}$, where $\widehat{D}$ is given by \red{``filling gaps" of diameter at most $\varepsilon$ in $D\subset \R$} as defined above.

Notice that for any such function $\xi(\cdot)$ \red{(that does not see gaps of diameter at most $\varepsilon$ and satisfies the measurability assumption from Theorem~\ref{thm:main2})} 
\begin{align*}
\tilde{\mathscr{K}}^s_z\xi=\tilde \kappa^s_z\xi
\end{align*}
holds, where $\tilde \kappa^s_z\xi$ is the function defined by \eqref{def:contKS-op}. \red{In particular, the left hand side of the equation is well-defined.} Since the functions $\tilde T_N(\cdot;z)$, $N\in\N$, and $\tilde T(\cdot;z)$ \red{do not see gaps of diameter at most $\varepsilon$} (by our assumption $\varepsilon<\delta$ and the respective definitions), Proposition \ref{prop:cavityint} implies
	$$
		\tilde T\bigl(\cdot;z\bigr) = e(\cdot) +\tilde{\mathscr K}_z^s \tilde T(\cdot;z).
	$$
and 
	$$
		\tilde T_{N+1}\bigl(\cdot;z\bigr) = e(\cdot) +  (\tilde{\mathscr K}_z^s \tilde T_N)\bigl(\cdot;z\bigr).
	$$
Assumption \eqref{eq:cont_selection} is equivalent to $e(D)+(\tilde{\mathscr K}^s_z\e^a)(D)\leq\e^{a(D)}$ for any interval $D\subset \R$. We prove by induction over $N$ that $\tilde T_N(D;z) \leq \e^{a(D)}$ for all $N\in \N$ and all intervals $D\subset \R$. For $N=1$, we have by our assumption
	$$
		\tilde T_1(D,z) = e(D)\leq e(D) + (\tilde{\mathscr K}_z^s \e^a)(D)\leq \e^{a(D)}
	$$
for all intervals $D\subset \R$.
	Next, assume for some $N\in \N$ that $\tilde T_N(D;z) \leq \e^{a(D)}$ for all intervals $D\subset\R$, then 
	$$
		{\tilde T}_{N+1}(D;z) = e(D)+\tilde{\mathscr K_z}^s \tilde{T}_{N}(D;z) \leq  e(D) +(\tilde{\mathscr  K}_z^s \e^a)(D)\leq \e^{a(D)},
	$$
where the first inequality holds by the inductive hypothesis, by monotonicity of $\tilde{\mathscr K}_z^s$ on non-negative functions and by the observation that for intervals $D\subset \R$ all the arguments of $\xi$ appearing in the definition of $(\tilde{\mathscr K}_z^s\xi)(D)$ are again intervals.

	This completes the induction and proves $T_N(D;z)\leq \e^{a(D)}$ for all $N\in\N$ and all intervals $D\subset\R$. Taking the limit $N\to\infty$ yields the corresponding bound for $\tilde{T}(D;z)$. The claim for arbitrary bounded subsets follows since every bounded subset is contained in some compact interval and
\begin{align*}
\tilde{T}(D_1;z)\leq\tilde T(D_2,z)
\end{align*}
for $D_1\subset D_2\subset \R$.
\end{proof}

\begin{proof}[Proof of Theorem~\ref{thm:cont1T}(a)]
	In analogy to the discrete case, for $\alpha>0$ and $L >0$ let 
	$$
		V_L(D):= \int_{-\infty}^\infty \1_{\{ [x,x+L]\cap D\neq \varnothing\}}\, \dd x.
	$$
	and $a(D)\equiv a_{\alpha,L}(D):= \alpha\, V_L(D)$. The choice of  $\alpha$ and $L$ is specified later in the proof. We apply Proposition~\ref{prop:cont_selection}  with the choice of the chopping map introduced at the beginning of this subsection and the selection rule $s$ that picks the leftmost snippet.

Remember the indicator $I(E_1;D'; Y_1,\ldots, Y_k)$ from Eq.~\eqref{eq:idef}. 
We show that there exists $\alpha >0$ such that 
\be \label{cytes}
	\sum_{k=1}^\infty \frac{1}{k!}\int_{\mathbb X^k} I(E_1;D';Y_1,\ldots,Y_k) \e^{\alpha [V_L(D'\cup Y_1\cup \cdots \cup Y_k) - V_L(D')]} \lambda_z^k(\dd\vect Y) 
		\leq \e^{\alpha V_L(D'\cup E_1) - V_L(D')} - 1
\ee
for all intervals $D'=[a,b)\subset \R$ or $D' = [a,b]$, including the empty set $D'= \varnothing$,  and all snippets $E_1 = [(k-1)\eps, a)\in \mathbb E_\eps$. 

If $D'$ is non-empty, then because of $\inf_{\ell\in \N} L_\ell \geq\eps$ and $|E_1|\leq \eps$ there cannot be two or more disjoint rods in $X$ that intersect $E_1$ but do not intersect $D'$, so the inequality to be proven reduces to 
\be\label{cytes2}
	\int_{\mathbb{X}} \e^{\alpha (V_L(D'\cup Y)-V_L(D'))} \mathbbm{1}_{\{Y\cap E_1\neq \varnothing,\ Y\cap D'= \varnothing\}}\lambda_z(\dd Y)\leq \e^{\alpha(V_L(D'\cup E_1)-V_L(D'))} - 1. 
\ee
Assuming that $L \geq \eps$, this is equivalent to
\be \label{cytes3}
	\sum_{\ell=1}^\infty z_\ell \int_0^{|E_1|}\, \e^{\alpha (L_\ell+x)} \dd x \leq \e^{\alpha|E_1|} - 1. 
\ee
The integral on the left-hand side is equal to $\exp( \alpha L_\ell) [\exp(\alpha |E_1|)- 1]/\alpha$, so we find that~\eqref{cytes} is equivalent to 
$$
	\sum_{\ell=1}^\infty z_\ell\, \e^{\alpha L_\ell} \leq \alpha,
$$
which holds true because of the assumption~\eqref{c1dTsuff}.

If $D'$ is empty, we note that there can be at most two disjoint rods in $X$ that intersect the snippet $E_1$, hence~\eqref{cytes} becomes 
\begin{multline} \label{cytes4}
	\int_\mathbb X \e^{\alpha V_L(Y)} \1_{\{Y\cap E_1\neq \varnothing\}} \lambda_z(\dd Y) 
	+ \frac12 \int_{\mathbb X^2} \e^{\alpha V_L(Y_1\cup Y_2)} \1_{\{Y_1\cap E_1\neq \varnothing,\, Y_2\cap E_1\neq \varnothing,\, Y_1\cap Y_2 = \varnothing\}} \lambda_z^2\bigl(\dd (Y_1,Y_2)\bigr) \\
	\leq \e^{\alpha V_L(E_1)} -1. 
\end{multline} 
The right-hand side is equal to $\exp(\alpha (L + |E_1|)) -1$. The first term on the left-hand side is equal to 
$$
	\sum_{\ell=1}^\infty z_\ell (L_\ell + |E_1|) \e^{\alpha (L + L_\ell)} = 	\sum_{\ell=1}^\infty z_\ell\, L_\ell \,  \e^{\alpha (L + L_\ell)} + O(\eps).
$$
The second term on the left-hand side of~\eqref{cytes4} is equal to 
$$
	\sum_{\ell,r=1}^\infty z_{\ell} z_r \int_{E_1^2}  \1_{\{x<y\}}\, \e^{\alpha V_L( [x - L_\ell, y+L_r])} \dd x \dd y
$$
which is bounded by 
$$
	\Bigl(\sum_{\ell=1}^\infty z_\ell \e^{\alpha L_\ell} \Bigr)^2 \e^{\alpha (\eps+L)}|E_1|^2 = O(\eps^2).
$$
For the inequality~\eqref{cytes4} to be satisfied, it is sufficient that 
\be \label{cytes5}
	\sum_{\ell=1}^\infty L_\ell z_\ell  \e^{\alpha L_\ell} + O(\eps) \leq \e^{\alpha |E_1|} - \e^{-\alpha L}. 
\ee
Arguments similar to the proof of  Theorem~\ref{Tres}(b), applied to the convex function $h:\R_+\to \R$, $h(u): = 1+\sum_{\ell=1}^\infty z_\ell u^{L_\ell}$, show that under condition~\eqref{c1dTsuff} there exists $\alpha >0$ such that not only condition~\eqref{c1dTsuff} holds true but in addition 
$$
	h'(\e^\alpha) =	\sum_{\ell=1}^\infty L_\ell z_\ell \e^{\alpha L_\ell} < 1. 
$$
Thus one can choose $L = L(\alpha)$ large enough and $\eps$ small enough so that~\eqref{cytes5} and hence~\eqref{cytes4} hold true. 
\end{proof}

\begin{proof}[Proof of Theorem~\ref{thm:cont1T}(b)]
We proceed as in the proof of Theorem~\ref{Tres}(b). Suppose that the expansions are absolutely convergent and define 
$$
	a(D):= \log T(D;-z) = \sum_{k=1}^\infty \frac{1}{k!}\int_{\mathbb X^k} \1_{\{ \exists i\in [n]:\, Y_i \cap D \neq \varnothing\}} \bigl|\varphi_{k}^\mathsf T(Y_1,\ldots, Y_k)\bigr|\, \lambda_z^k(\dd \vect Y). 
$$
Then \blue{by Proposition 3.13 and since $\inf_{\ell\in \N} L_\ell$ is bounded from below by $
\eps>0$}, 
\be \label{eq:contnec1}
	\int_{\mathbb X} \1_{\{ X\cap E_1 \neq \varnothing,\, X\cap D'= \varnothing\}} \, \e^{a(D'\cup X) - a(D')} \lambda_z(\dd X)  \leq \e^{a(D'\cup E_1) - a(D')} - 1
\ee
for example for $E_1= [0,\eps)$ and $D' = [\eps,\eps+L]$ with $L>0$ and $\eps$ sufficiently small. 

Before we evaluate the two sides of the inequality, we note two useful properties of $a(\cdot)$. First, \red{the map $a$ does not see gaps of diameter at most $\eps$}. Precisely, if $X=[x- L_\ell, x]$ with $x\in [0,\eps)$ and $D'$ is as above, then 
$$
 a(D'\cup X) = a([x-L_\ell,\eps+L]).
$$
Indeed, any rod $Y_i\in \mathbb X$ that intersects $[0,\eps)$ must also intersect $D'\cup X$ because it has a length $|Y_i| \geq \eps$. 
Second, because of translational invariance, the weight $a(D)$ of a non-empty  interval depends only on its length $|D|$. We check that in addition, it is an affine function of the length. For $x\in \R$, define
$$
	\alpha(x):= \sum_{\ell=1}^\infty z_\ell \sum_{k=0}^\infty \frac{1}{k!} \int_{\mathbb X^k} \1_{\{\forall i \in [k]:\, Y_i \subset (-\infty, x]\}}  \bigl|\varphi_{1+k}^\mathsf T(Y_1,\ldots, Y_k, [x-L_\ell,x])\bigr|\, \lambda_z^k(\dd \vect Y).
$$
The quantity $\alpha(x)$ is best thought of as an integral over clusters in which  the right-most rod $[x-L_\ell, x]$  has its right end pinned at $x$. By translational invariance, $\alpha(x)$ is actually independent of $x$ and we may write $\alpha(x) \equiv \alpha$ for some scalar $\alpha \geq 0$. 
Now let $I= [a,b]$ and $J= [b,c]$ with $a<b<c$. Then 
$$
	a(I\cup J) - a(J) = \sum_{k=1}^\infty \frac{1}{k!}\int_{\mathbb X^k} \1_{\{ \exists i\in [k]:\, Y_i \cap I \neq \varnothing\}}\1_{\{\forall i\in [k]:\, Y_i \cap J = \varnothing\}} \bigl|\varphi_{k}^\mathsf T(Y_1,\ldots, Y_k)\bigr|\, \lambda_z^k(\dd \vect Y).
$$
Any cluster $(Y_1,\ldots, Y_k)$ that intersects $I$ but not $J$ has its right-most end in $[a,b)$, therefore 
$$
	a(I\cup J) - a(J) = \int_I \alpha(x) \dd x = \alpha\, |I|. 
$$
With these two observations, the left-hand side of~\eqref{eq:contnec1} becomes 
$$
	\sum_{\ell=1}^\infty z_\ell \int_0^\eps \e^{ a( [x- L_\ell,x]\cup D') - a(D')} \dd x
		=\sum_{\ell=1}^\infty z_\ell \int_0^\eps \e^{\alpha (x+ L_\ell)} \dd x 
			= \sum_{\ell=1}^\infty z_\ell \e^{\alpha L_\ell} \frac{1}{\alpha}(\e^{\alpha \eps} - 1)
$$
while the right-hand side of~\eqref{eq:contnec1} is $\exp(\alpha \eps) - 1$. It follows that 
\begin{equation*}
	 \sum_{\ell=1}^\infty z_\ell \e^{\alpha L_\ell} \leq \alpha.  \qedhere
\end{equation*} 
\end{proof}

\appendix 

\section{Proof of Lemma~\ref{kl2}}\label{AppB}
\begin{proof}[Proof of Lemma \ref{kl2}]\label{proof:k12}
We show that the system of inequalities
\begin{align}\label{eqkl3}
\frac{1+\sum\limits_{k\geq 1}\sum\limits_{{\substack{Y=\{y_1,\ldots,y_k\}\\y_i \nsim x_1,\ y_i\sim y_j}}}\prod\limits_{i=1}^k\mu(y_i)\prod\limits_{w\in\Gamma(Y)}\e^{\mu(w)}}{\prod\limits_{q\in\Gamma(x_1)\cap\Gamma(X)}\e^{\mu(q)}}\geq 1+\sum\limits_{k\geq 1}\sum\limits_{\substack{\substack{Y=\{y_1,...,y_k\}\\y_i \nsim x_1,\ y_i\sim y_j}\\y_i\sim X}}\prod\limits_{i=1}^k\mu(y_i)\prod\limits_{w\in\Gamma(Y)\cap\Gamma(X)^C}\e^{\mu(w)},
\end{align}
which is equivalent to \eqref{eqkl2}, holds under the assumptions of the lemma. We do so by proving the following three claims. However, first we would like to introduce some additional notation to complement the notation from Section~\ref{sec:abstract-polymers}. \\\\
 For given $x_1\in\mathbb{X}$ and $X=\{x_2,...,x_p\}\subset \mathbb{X}$ let $Q$ denote the set $\Gamma(x_1)\cap\Gamma(X)$ and let $\mathcal C$ denote the set of (non-empty) compatible subsets of $Q$. Furthermore, we define the family $(A_U)_{U\subset Q}$, $A_U=A_U(x_1,Q,\mu)$, indexed by all the subsets $U\subset Q$ (including the empty set), by
\begin{align*}
A_U=A_U(x_1,Q,\mu):=\sum\limits_{k\geq 1}\sum\limits_{Y=\{y_1,...,y_k\}}\prod\limits_{i=1}^k\mu(y_i)\prod\limits_{w\in\Gamma(Y)\backslash U}\e^{\mu(w)},
\end{align*}
where the sum is over subsets $Y=\{y_1,...,y_k\}\subset\mathbb{X}$ such that the following constraints are satisfied: $Y$ is an compatible set, $Y\subset \Gamma(x_1)$, $Y\cap Q=\varnothing$ and $U= \Gamma(Y)\cap Q$.\\
Finally, define the family of coefficients $(\beta_U)_{U\subset Q}$, $\beta_U=\beta_U(x_1,Q,\mu)$, also indexed by all the subsets $U\subset Q$ (including the empty set), by
\begin{align*}
\beta_U=\beta_U(x_1,Q,\mu):=\prod\limits_{q\in Q\backslash U}\e^{-\mu(q)}+\sum\limits_{\substack{C\in\mathcal C\\C\cap U=\varnothing}}\prod\limits_{c\in C}\mu(c)\prod\limits_{w\in (Q\backslash U)\backslash \Gamma(C)}\e^{-\mu(w)}.
\end{align*}
Then the following statements hold true:
\begin{claim}\label{claimB1} The right-hand side of \eqref{eqkl3} is bounded from above by 
\begin{align*}
1+\sum_{U\subset Q}A_U.
\end{align*}
\end{claim}
\begin{claim}\label{claimB2}
The left-hand side of \eqref{eqkl3} is bounded from below by
\begin{align}\label{idlhs}
\beta_{\varnothing}+\sum\limits_{U\subset Q}\beta_UA_U.
\end{align}
\end{claim}
\begin{claim}\label{claimB3} The lower bounds $\beta_U\geq 1$ hold for every $U\subset Q$ and thus 
\begin{align*}
\beta_{\varnothing}+\sum\limits_{U\subset Q}\beta_UA_U\geq 1+\sum_{U\subset Q}A_U.
\end{align*}
\end{claim}

All the sums in the three claims above run over subsets of $Q=\Gamma(x_1)\cap\Gamma(X$) including the empty set (so that the coefficient $\beta_\varnothing$ appears in~\eqref{idlhs} twice). The inequalities \eqref{eqkl3} follow directly from the three claims. Proving the claims is thus sufficient to conclude the proof of the lemma:
\begin{proof}[Proof of Claim~\ref{claimB1}]Reorder the sum in the right-hand side of  \eqref{eqkl3} by putting the summand together which belong to the same $U:=\Gamma(Y)\cap Q$. Notice that the constraint $Y\sim X$ implies $Y\cap Q=\varnothing$ since $Y\subset \Gamma(x_1)$. The claim now follows directly from the simple observation that $\Gamma(Y)\cap \Gamma(X)^C\subset\Gamma(Y)\backslash U$ for any Y and thus \[\prod\limits_{w\in\Gamma(Y)\backslash U}\e^{\mu(w)}\geq\prod\limits_{w\in\Gamma(Y)\cap \Gamma(X)^C}\e^{\mu(w)}.\]
\end{proof}
\begin{proof}[Proof of Claim~\ref{claimB2}] To see that the bounds stated in the claim hold, decompose the sum in the left-hand side of \eqref{eqkl3} as 
\begin{align}\label{eq:reord1}
&\nonumber 1+\sum\limits_{k\geq 1}\sum\limits_{{\substack{Y=\{y_1,...,y_k\}\\y_i \nsim x_1,\ y_i\sim y_j}}}\prod\limits_{i=1}^k\mu(y_i)\prod\limits_{w\in\Gamma(Y)}\e^{\mu(w)}\\=&1+\sum\limits_{C\in\mathcal C\cup\{\varnothing\}}\sum\limits_{k\geq 1}\sum\limits_{{\substack{Y=\{y_1,...,y_k\}\\y_i \nsim x_1,\ y_i\sim y_j}}}\mathbbm{1}_{\{C=Q\cap Y\}}\prod\limits_{i=1}^k\mu(y_i)\prod\limits_{w\in\Gamma(Y)}\e^{\mu(w)}.
\end{align}
Notice that for any $C\in\mathcal C$
\begin{align*}
&\sum\limits_{k\geq 1}\sum\limits_{{\substack{Y=\{y_1,\ldots,y_k\}\\y_i \nsim x_1,\ y_i\sim y_j}}}\mathbbm{1}_{\{C=Q\cap Y\}}\prod\limits_{i=1}^k\mu(y_i)\prod\limits_{w\in\Gamma(Y)}\e^{\mu(w)}\\ \geq&\prod\limits_{c\in C}\mu(c)\prod\limits_{w\in \Gamma(C)\cap Q}\e^{\mu(w)}\left(1+\sum\limits_{k\geq 1}\sum\limits_{{\substack{Y=\{y_1,...,y_k\}\\y_i \nsim x_1,\ y_i\sim y_j}}}\mathbbm{1}_{\{Y\cap Q =\varnothing\}}\mathbbm{1}_{\{Y\sim C\}}\prod\limits_{i=1}^k\mu(y_i)\prod\limits_{w\in\Gamma(Y)\backslash (\Gamma(C)\cap Q)}\e^{\mu(w)}\right).
\end{align*}
This estimate is established by discarding the exponential weights corresponding to a subset of $\Gamma(C)$: Multiplying the lower bound in the last line with \[\prod\limits_{w\in \Gamma(C)\backslash \Gamma(Y)\backslash Q}\e^{\mu(w)}\geq 1\] yields equality. No further estimates are necessary to prove the claim; we simply plug the obtained lower bound into the right-hand side of~\eqref{eq:reord1} and get
\begin{align*}
&1+\sum\limits_{k\geq 1}\sum\limits_{{\substack{Y=\{y_1,...,y_k\}\\y_i \nsim x_1,\ y_i\sim y_j}}}\mathbbm{1}_{\{Y\cap Q =\varnothing\}}\prod\limits_{i=1}^k\mu(y_i)\prod\limits_{w\in\Gamma(Y)}\e^{\mu(w)}+\sum\limits_{C\in\mathcal C}\prod\limits_{c\in C}\mu(c)\prod\limits_{w\in \Gamma(C)\cap Q}\e^{\mu(w)}
\\&\times\left(1+\sum\limits_{k\geq 1}\sum\limits_{{\substack{Y=\{y_1,...,y_k\}\\y_i \nsim x_1,\ y_i\sim y_j}}}\mathbbm{1}_{\{Y\cap Q =\varnothing\}}\mathbbm{1}_{\{Y\sim C\}}\prod\limits_{i=1}^k\mu(y_i)\prod\limits_{w\in\Gamma(Y)\backslash (\Gamma(C)\cap Q)}\e^{\mu(w)}\right)
\end{align*}
or, equivalently,
\begin{align*}
1&+\sum\limits_{C\in\mathcal C}\prod\limits_{c\in C}\mu(c)\prod\limits_{w\in \Gamma(C)\cap Q}\e^{\mu(w)}+\sum\limits_{k\geq 1}\sum\limits_{{\substack{Y=\{y_1,...,y_k\}\\y_i \nsim x_1,\ y_i\sim y_j}}}\mathbbm{1}_{\{Y\cap Q =\varnothing\}}\prod\limits_{i=1}^k\mu(y_i)\prod\limits_{w\in\Gamma(Y)}\e^{\mu(w)}\\&+\sum\limits_{C\in\mathcal C}\prod\limits_{c\in C}\mu(c)\prod\limits_{w\in \Gamma(C)\cap Q}\e^{\mu(w)}\sum\limits_{k\geq 1}\sum\limits_{{\substack{Y=\{y_1,...,y_k\}\\y_i \nsim x_1,\ y_i\sim y_j}}}\mathbbm{1}_{\{Y\cap Q =\varnothing\}}\mathbbm{1}_{\{Y\sim C\}}\prod\limits_{i=1}^k\mu(y_i)\prod\limits_{w\in\Gamma(Y)\backslash (\Gamma(C)\cap Q)}\e^{\mu(w)}
\end{align*}
as a lower bound for the left-hand side of \eqref{eq:reord1}.

Reordering the last expression by summing over $Y$ first, one realizes that it is equal to
\begin{align*}
&1+\sum\limits_{C\in\mathcal C}\prod\limits_{c\in C}\mu(c)\prod\limits_{w\in \Gamma(C)\cap Q}\e^{\mu(w)}+\sum\limits_{k\geq 1}\sum\limits_{{\substack{Y=\{y_1,...,y_k\}\\y_i \nsim x_1,\ y_i\sim y_j}}}\mathbbm{1}_{\{Y\cap Q =\varnothing\}}\prod\limits_{i=1}^k\mu(y_i)\prod\limits_{w\in\Gamma(Y)}\e^{\mu(w)}
\\
&\times\left(1+\sum\limits_{\substack{C\in\mathcal C\\C\sim Y}}\prod\limits_{c\in C}\mu(c)\prod\limits_{w\in (\Gamma(C)\cap Q)\backslash\Gamma (Y)}\e^{\mu(w)}\right),
\end{align*}
which may be rewritten as
\begin{align*}
&1+\sum\limits_{C\in\mathcal C}\prod\limits_{c\in C}\mu(c)\prod\limits_{w\in \Gamma(C)\cap Q}\e^{\mu(w)}
+\sum\limits_{k\geq 1}\sum\limits_{{\substack{Y=\{y_1,...,y_k\}\\y_i \nsim x_1,\ y_i\sim y_j}}}\mathbbm{1}_{\{Y\cap Q =\varnothing\}}\prod\limits_{i=1}^k\mu(y_i)\prod\limits_{w\in\Gamma(Y)\backslash(\Gamma(Y)\cap Q)}\e^{\mu(w)}
\\
&\times\left(\prod\limits_{w\in \Gamma(Y)\cap Q}\e^{\mu(w)}+\sum\limits_{\substack{C\in\mathcal C\\C\sim Y}}\prod\limits_{c\in C}\mu(c)\prod\limits_{w\in (\Gamma(C)\cup\Gamma(Y))\cap Q}\e^{\mu(w)}\right).
\end{align*}

From the last expression, we obtain precisely  the sum in~\eqref{idlhs} by putting the summands in the last expression which belong to the same $U=\Gamma(Y)\cap Q$ together and dividing by $\prod_{q\in Q}\e^{\mu(q)}$. This yields the claimed lower bound.
\end{proof}
\begin{proof}[Proof of Claim~\ref{claimB3}]
Without loss of generality, we may assume that $Q=\Gamma(x_1)\cap\Gamma(X)$ is a finite non-empty set. Then the claim can be proven via induction over the cardinality of $Q$. 

To start the induction consider the case $Q=\{q\}$, $q\in \mathbb{X}$.  Then $\beta_\varnothing=\e^{-\mu(q)}+\mu(q)\geq 1$ and $\beta_{\{q\}}=1$ by definition.

For the inductive step, let $n\in\mathbb{N}$, $Q_n=\{q_1,...,q_n\}\subset\mathbb{X}$ and let $\mathcal C_n$ be the set of compatible subsets of $Q_n$. Furthermore, let $q_{n+1}\in \mathbb{X}\backslash Q_n$ and let $Q_{n+1}=Q_n\cup \{q_{n+1}\}$. Naturally, there exists a family of subsets $\Theta_n\subset\mathcal{C}_n$, such that the set $\mathcal{C}_{n+1}$ of compatible subsets of $Q_{n+1}$ is  given by  $\mathcal{C}_{n+1}=\{C\cup\{q_{n+1}\}\vert C\in \Theta_n \textit{ or }C=\varnothing\}\cup\mathcal C_n=:\overline{\Theta}_n\cup \mathcal C_n$. 

Under the assumption that $\beta_U(Q_n)\geq 1$ for all $U\subset Q_n$ it is to show that $\beta_U(Q_{n+1})\geq 1$ for all $U\subset Q_{n+1}$. Therefore let $U\subset Q_{n+1}$. If $q_{n+1}\in U$ then $\beta_U(Q_{n+1})= \beta_{U\backslash\{q_{n+1}\}}(Q_n)\geq 1$ by the inductive hypothesis. Left to consider is the case $q_{n+1}\notin U$ (and thus $U\subset Q_n$). Recall that  we defined the coefficient $\beta_U(Q_{n+1})$ by
\begin{align*}
\beta_U(Q_{n+1})=\prod\limits_{q\in Q_{n+1}\backslash U}\e^{-\mu(q)}+\sum\limits_{\substack{C\in\mathcal C_{n+1}\\C\cap U=\varnothing}}\prod\limits_{c\in C}\mu(c)\prod\limits_{w\in (Q_{n+1}\backslash U)\backslash \Gamma(C)}\e^{-\mu(w)}.
\end{align*}
Using the decomposition $\mathcal C_{n+1}=\overline{\Theta}_n\cup \mathcal C_n$, we get
\begin{align*}
\beta_U(Q_{n+1})=&\ \e^{-\mu(q_{n+1})}\prod\limits_{q\in Q_{n}\backslash U}\e^{-\mu(q)}+\sum\limits_{\substack{C\in\overline{\Theta}_n\\C\cap U=\varnothing}}\prod\limits_{c\in C}\mu(c)\prod\limits_{w\in (Q_{n+1}\backslash U)\backslash \Gamma(C)}\e^{-\mu(w)}
\\
&+\sum\limits_{\substack{C\in\mathcal C_n\\C\cap U=\varnothing}}\prod\limits_{c\in C}\mu(c)\prod\limits_{w\in (Q_{n+1}\backslash U)\backslash \Gamma(C)}\e^{-\mu(w)}
\\
=&\ \e^{-\mu(q_{n+1})}\prod\limits_{q\in Q_{n}\backslash U}\e^{-\mu(q)}+\mu(q_{n+1})\prod\limits_{q\in (Q_{n}\backslash U)\backslash\Gamma(q_{n+1})}\e^{-\mu(q)}
\\
&+\mu(q_{n+1})\sum\limits_{\substack{C\in\Theta_n\\C\cap U=\varnothing}}\prod\limits_{c\in C}\mu(c)\prod\limits_{w\in (Q_{n+1}\backslash U)\backslash \Gamma(C\cup\{q_{n+1}\})}\e^{-\mu(w)}\\&+\e^{-\mu(q_{n+1})}\sum\limits_{\substack{C\in\Theta_n\\C\cap U=\varnothing}}\prod\limits_{c\in C}\mu(c)\prod\limits_{w\in (Q_{n}\backslash U)\backslash \Gamma(C)}\e^{-\mu(w)}
\\&+\sum\limits_{\substack{C\in\mathcal C_n\backslash\Theta_n\\C\cap U=\varnothing}}\prod\limits_{c\in C}\mu(c)\prod\limits_{w\in (Q_{n}\backslash U)\backslash \Gamma(C)}
\e^{-\mu(w)}.
\end{align*}
Multiplying the second summand in the last expression by the products of negative exponential weights \[\prod\limits_{w\in\Gamma(\{q_{n+1}\})}\e^{-\mu(w)})\leq 1\]and the third summand by \[\qquad\prod\limits_{w\in\Gamma(\{q_{n+1}\})\backslash\Gamma(C)}\e^{-\mu(w)})\leq 1,\] we obtain the following lower bound:
\begin{align*}
&\beta_U(Q_{n+1})
\\
&\geq \left(\e^{-\mu(q_{n+1})}+\mu(q_{n+1})\right)\left(\ \prod\limits_{q\in Q_{n}\backslash U}\e^{-\mu(q)} +\sum\limits_{\substack{C\in\Theta_n\\C\cap U=\varnothing}}\prod\limits_{c\in C}\mu(c)\prod\limits_{w\in (Q_{n}\backslash U)\backslash \Gamma(C)}\e^{-\mu(w)}\right)\\&\quad+\sum\limits_{\substack{C\in\mathcal C_n\backslash\Theta_n\\C\cap U=\varnothing}}\prod\limits_{c\in C}\mu(c)\prod\limits_{w\in (Q_{n}\backslash U)\backslash \Gamma(C)}\e^{-\mu(w)}
\end{align*}
Since $\e^{-\mu(q_{n+1})}+\mu(q_{n+1})\geq 1$ for any $\mu$, this last expression is in turn bounded from below by 
\begin{align*}
&\prod\limits_{q\in Q_{n}\backslash U}\e^{-\mu(q)} +\sum\limits_{\substack{C\in\Theta_n\\C\cap U=\varnothing}}\prod\limits_{c\in C}\mu(c)\prod\limits_{w\in (Q_{n}\backslash U)\backslash \Gamma(C)}\e^{-\mu(w)}+\sum\limits_{\substack{C\in\mathcal C_n\backslash\Theta_n\\C\cap U=\varnothing}}\prod\limits_{c\in C}\mu(c)\prod\limits_{w\in (Q_{n}\backslash U)\backslash \Gamma(C)}\e^{-\mu(w)}\\
&=\prod\limits_{q\in Q_{n}\backslash U}\e^{-\mu(q)} +\sum\limits_{\substack{C\in\mathcal C_n\\C\cap U=\varnothing}}\prod\limits_{c\in C}\mu(c)\prod\limits_{w\in (Q_{n}\backslash U)\backslash \Gamma(C)}\e^{-\mu(w)}=\beta_U(Q_n).
\end{align*}
By the inductive hypothesis $\beta_U(Q_n)$ is bounded from below by $1$, hence we have shown $\beta_U(Q_{n+1})\geq 1$. This concludes the induction and therefore also the proof of Claim~\ref{claimB3}.
\end{proof}
Combining the three statements from the claims \ref{claimB1},~\ref{claimB2} and~\ref{claimB3} immediately yields the claim of Lemma~\ref{kl2}.
\end{proof}
\section*{Competing interests}
The authors have no competing interests to declare that are relevant to the content of this article.
\section*{Funding}
The authors did not receive support from any funding body for the submitted work.
\section*{Data availability}
Data sharing not applicable to this article as no datasets were generated or analysed during the current study.
\bibliographystyle{plain}
\bibliography{clusterbiblio}

\end{document}